\documentclass[a4paper,UKenglish]{lipics-v2016}
 
\usepackage{microtype}
\usepackage{xspace,cleveref}

\crefname{theorem}{theorem}{\bf Theorem}
\crefname{observation}{observation}{\bf Observation}
\crefname{lemma}{lemma}{\bf Lemma}
\crefname{corollary}{corollary}{\bf Corollary}
\crefname{proposition}{proposition}{\bf Proposition}
\crefname{definition}{definition}{\bf Definition}
\crefname{claim}{claim}{\bf Claim}
\crefname{reductionrule}{reduction rule}{\bf Reduction rule}
\crefname{chapter}{chapter}{\bf Chapter}

\usepackage{thmtools,thm-restate}
\newtheorem{observation}{\bf Observation}

\newcommand{\longversion}[1]{}
\newcommand{\shortversion}[1]{#1}

\newcommand{\el}{\ensuremath{\ell}\xspace}
\newcommand{\suc}{\ensuremath{\succ}\xspace}

\usepackage{makecell}
\usepackage{multirow}
\usepackage{mathtools}
\renewcommand{\leq}{\leqslant}
\renewcommand{\geq}{\geqslant}
\renewcommand{\ge}{\geqslant}
\renewcommand{\le}{\leqslant}

\newcommand{\YES}{\textsc{Yes}\xspace}
\newcommand{\NO}{\textsc{No}\xspace}

\newcommand{\pr}{\ensuremath{\prime}\xspace}
\newcommand{\prr}{{\ensuremath{\prime\prime}}\xspace}

\newcommand{\PW}{\textsc{Possible Winner}\xspace}
\newcommand{\TDM}{\textsc{Three Dimensional Matching}\xspace}
\newcommand{\MIS}{\textsc{Multicolored Independent Set}\xspace}

\newcommand{\SAT}{\ensuremath{(3,B2)}--\textsc{SAT}\xspace}

\newcommand{\NP}{\ensuremath{\mathsf{NP}}\xspace}
\newcommand{\Pb}{\ensuremath{\mathsf{P}}\xspace}

\newcommand{\NPC}{\ensuremath{\mathsf{NP}}-complete\xspace}

\newcommand{\NB}{\ensuremath{\mathbb N}\xspace}

\renewcommand{\AA}{\ensuremath{\mathcal A}\xspace}
\newcommand{\BB}{\ensuremath{\mathcal B}\xspace}
\newcommand{\CC}{\ensuremath{\mathcal C}\xspace}
\newcommand{\DD}{\ensuremath{\mathcal D}\xspace}
\newcommand{\EE}{\ensuremath{\mathcal E}\xspace}

\newcommand{\GG}{\ensuremath{\mathcal G}\xspace}

\newcommand{\II}{\ensuremath{\mathcal I}\xspace}

\newcommand{\LL}{\ensuremath{\mathcal L}\xspace}

\newcommand{\PP}{\ensuremath{\mathcal P}\xspace}
\newcommand{\QQ}{\ensuremath{\mathcal Q}\xspace}

\renewcommand{\SS}{\ensuremath{\mathcal S}\xspace}
\newcommand{\TT}{\ensuremath{\mathcal T}\xspace}
\newcommand{\UU}{\ensuremath{\mathcal U}\xspace}
\newcommand{\VV}{\ensuremath{\mathcal V}\xspace}
\newcommand{\WW}{\ensuremath{\mathcal W}\xspace}
\newcommand{\XX}{\ensuremath{\mathcal X}\xspace}
\newcommand{\YY}{\ensuremath{\mathcal Y}\xspace}
\newcommand{\ZZ}{\ensuremath{\mathcal Z}\xspace}

\newcommand{\aaa}{\ensuremath{\mathfrak a}\xspace}
\newcommand{\bbb}{\ensuremath{\mathfrak b}\xspace}

\newcommand{\ppp}{\ensuremath{\mathfrak p}\xspace}
\newcommand{\qqq}{\ensuremath{\mathfrak q}\xspace}

\newcommand{\sss}{\ensuremath{\mathfrak s}\xspace}

\newcommand{\vvv}{\ensuremath{\mathfrak v}\xspace}

\newcommand{\nice}{smooth}
\usepackage{nicefrac}

\newcommand{\nfrac}{\nicefrac}

\usepackage[obeyDraft]{todonotes}

\newcounter{nmcomment}
\newcommand{\nmtodo}[1]{%
    \refstepcounter{nmcomment}%
    {%
        \todo[color={LightSteelBlue},size=\small,inline]{%
            \textbf{Comment [NM\thenmcomment ]:}~#1}%
}}

\newcounter{pdcomment}

\usepackage{amsthm,amsmath,amssymb}

\title{On the Exact Amount of Missing Information that makes Finding Possible Winners Hard}
\titlerunning{On the Exact Amount of Missing Information that makes Finding Possible Winners Hard} 

\author[1]{Palash Dey}
\author[2]{Neeldhara Misra}
\affil[1]{Indian Institute of Science, Bangalore\\
  \texttt{palash@csa.iisc.ernet.in}}
\affil[2]{Indian Institute of Technology, Gandhinagar\\
  \texttt{neeldhara.m@iitgn.ac.in}}
\authorrunning{N. Misra and P. Dey} 

\Copyright{Palash Dey and Neeldhara Misra}

\subjclass{F.2.0 Analysis Of Algorithms And Problem Complexity - General }
\keywords{Computational Social Choice, Dichotomy, NP-completeness, Maxflow, Voting}

\EventEditors{John Q. Open and Joan R. Acces}
\EventNoEds{2}
\EventLongTitle{42nd Conference on Very Important Topics (CVIT 2016)}
\EventShortTitle{CVIT 2016}
\EventAcronym{CVIT}
\EventYear{2016}
\EventDate{December 24--27, 2016}
\EventLocation{Little Whinging, United Kingdom}
\EventLogo{}
\SeriesVolume{42}
\ArticleNo{23}

\begin{document}

\maketitle

\begin{abstract}

We consider election scenarios with incomplete information, a situation that arises often in practice. There are several models of incomplete information and accordingly, different notions of outcomes of such elections. In one well-studied model of incompleteness, the votes are given by partial orders over the candidates. In this context we can frame the problem of finding a \textit{possible winner}, which involves determining whether a given candidate wins in at least one completion of a given set of partial votes for a specific voting rule. 

The \PW problem is well-known to be \NPC{} in general, and it is in fact known to be \NPC{} for several voting rules where the number of undetermined pairs in every vote is bounded only by some constant. In this paper, we address the question of determining precisely the smallest number of undetermined pairs for which the \PW problem remains \NPC{}. In particular, we find the exact values of $t$ for which the \PW problem transitions to being \NPC from being in \Pb, where $t$ is the maximum number of undetermined pairs in every vote. We demonstrate tight results for a broad subclass of scoring rules which includes all the commonly used scoring rules (such as plurality, veto, Borda, $k$-approval, and so on), Copeland$^\alpha$ for every $\alpha\in[0,1]$, maximin, and Bucklin voting rules. A somewhat surprising aspect of our results is that for many of these rules, the \PW problem turns out to be hard even if every vote has at most one undetermined pair of candidates.

\end{abstract}
\section{Introduction}

In many real life situations including multiagent systems, agents often need to aggregate 
their preferences and agree upon a common decision (candidate). 
Voting is an immediate natural tool in these 
situations. Common and classic applications of voting in multiagent systems include collaborative filtering and recommender systems~\cite{PennockHG00}, spam detection~\cite{Cohen}, computational biology~\cite{JacksonSA08}, winner determination in sports competition~\cite{BetzlerBN14} etc. We refer the readers to \cite{moulin2016handbook} for an elaborate treatment of computational voting theory.

Usually, in a voting setting, it is assumed that the votes are complete orders over the
candidates. However, due to many reasons, for example, lack of knowledge of voters 
about some candidates, a voter maybe indifferent between some pairs of candidates. 
Hence, it is both natural and important to 
consider scenarios where votes are partial orders over the candidates. When votes are only partial orders over the candidates, the winner cannot be determined with certainty since it depends on 
how these partial orders are extended to linear orders. 
This leads to a natural computational problem called the \PW problem~\cite{konczak2005voting}:  given a set of partial votes $\PP$ and a distinguished candidate $c$, is there a way to extend the partial votes to complete votes where $c$ wins? The \PW problem has been studied extensively in the literature \cite{lang2007winner,pini2007incompleteness,walsh2007uncertainty,xia2008determining,betzler2009multivariate,chevaleyre2010possible,betzler2010partial,baumeister2011computational,lang2012winner,faliszewski2014complexity} following its definition in \cite{konczak2005voting}. Betzler et al.~\cite{betzler2009towards} and Baumeister et al.~\cite{BaumeisterR12} show that the \PW winner problem is \NPC for all scoring rules except for the plurality and veto voting rules; the \PW winner problem is in \Pb for the plurality and veto voting rules.
The \PW problem is known to be \NPC for many common voting rules, for example, a class of scoring rules, maximin, 
Copeland, Bucklin etc. even when the maximum number of undetermined pairs of candidates in every vote is bounded above by small constants~\cite{xia2008determining}. Walsh showed that the \PW problem can be solved in polynomial time for all the voting rules mentioned above when we have a constant number of candidates \cite{walsh2007uncertainty}.

\subsection{Our Contribution}

Our main contribution lies in pinning down exactly the minimum number of undetermined pairs allowed per vote so that the \PW winner problem continues to be \NPC for a large class of scoring rules, Copeland$^\alpha$, maximin, and Bucklin voting rules. To begin with, we describe our results for scoring rules. We work with a class of scoring rules that we call \nice{}, which are essentially scoring rules where the score vector for $(m+1)$ candidates can be obtained by either duplicating an already duplicated score in the score vector for $m$ candidates, or by extending the score vector for $m$ candidates at one of the endpoints with an arbitrary new value. While less general than the class of pure scoring rules, the \nice{} rules continue to account for all commonly used scoring rules (such as Borda, plurality, veto, $k$-approval, and so on). Using $t$ to denote the maximum number of undetermined pairs of candidates in every vote, we show the following.

\begin{itemize}
 \item The \PW problem is \NPC even when $t\le 1$ for scoring rules which have two distinct nonzero differences between consecutive coordinates in the score vector (we call them differentiating) and in \Pb when $t\le 1$ for other scoring rules [\Cref{thm:one-missing-pair}].
 
 \item Else the \PW problem is \NPC when $t\ge 2$ and in \Pb when $t\le 1$ for scoring rules that contain $(\alpha+1, \alpha+1, \alpha)$ for any $\alpha\in\NB$ [\Cref{thm:two-missing-pairs}].
 
 \item Else the \PW problem is \NPC when $t\ge 3$ and in \Pb when $t\le 2$ for scoring rules which contain $(\alpha+2, \alpha+1, \alpha+1, \alpha)$ for any $\alpha\in\NB$ [\Cref{thm:three-missing-pairs}].
 
 \item The \PW problem is \NPC when $t\ge 4$ and in \Pb when $t\le 3$ for $k$-approval and $k$-veto voting rules for any $k>1$ [\Cref{thm:four-missing-pairs}].
 
 \item The \PW problem is \NPC when $t\ge m-1$ and in \Pb when $t\le m-2$ for the scoring rule $(2, 1, 1, \ldots, 1, 0)$ [\Cref{thm:four-missing-pairs}].
\end{itemize}

We summarize our results for the Copeland$^\alpha$, maximin, and Bucklin voting rules in \Cref{tbl:summary_pw_exact_except_scoring_rules}. We observe that the \PW problem for the Copeland$^\alpha$ voting rule is \NPC even when every vote has at most $2$ undetermined pairs of candidates for $\alpha\in\{0, 1\}$. However, for $\alpha\in(0,1)$, the \PW problem for the Copeland$^\alpha$ voting rule is \NPC even when every vote has at most $1$ undetermined pairs of candidates.
Our results show that the \PW winner problem continues to be \NPC for all the common voting rules studied here (except $k$-approval) even when the number of undetermined pairs of candidates per vote is at most $2$. Other than finding the exact number of undetermined pairs needed per vote to make the \PW problem \NPC for common voting rules, we also note that all our proofs are much simpler and shorter than most of the corresponding proofs from the literature subsuming the work in \cite{xia2008determining,betzler2009towards,BaumeisterR12}.

\begin{table}[!htbp]
 \centering
 \renewcommand*{\arraystretch}{1.3}
 \begin{tabular}{|c|c|c|c|}\hline
 Voting rules & \NPC & Poly time & Known from literature~\cite{xia2008determining}\\\hline\hline
 
 Copeland$^{0,1}$ & $t\ge2$ [\Cref{thm:copeland_hard_two}] & $t\le 1$ [\Cref{thm:copeland_poly_one}] & \multirow{3}{*}{\NPC for $t\ge 8$} \\\cline{1-3}
 
 \makecell{Copeland$^\alpha$\\$\alpha\in(0,1)$} & $t\ge1$ [\Cref{thm:copeland_hard_one}] & -- &  \\\hline
 
 Maximin & $t\ge2$ [\Cref{thm:maximin_hard}] & $t\le 1$ [\Cref{thm:maximin_poly_one}] & \NPC for $t\ge 4$ \\\hline
 
 Bucklin & $t\ge2$ [\Cref{thm:bucklin_hard}] & $t\le 1$ [\Cref{thm:bucklin_hard}] & \NPC for $t\ge 16^\star$ \\\hline
 \end{tabular}
 \caption{Summary and comparison of results from the literature for Copeland$^\alpha$, maximin, and Bucklin voting rules. $^\star$The result was proved for the simplified Bucklin voting rule but the proof can be modified easily for the Bucklin voting rule.}\label{tbl:summary_pw_exact_except_scoring_rules}
\end{table}
\section{Preliminaries}

Let us denote the set $\{1, 2, \ldots, n\}$ by $[n]$ for any positive integer $n$. Let $\CC = \{c_1, c_2, \ldots, c_m\}$ be a set of candidates or alternatives and $\VV = \{v_1, v_2, \ldots, v_n\}$ a set of voters. If not mentioned otherwise, we denote the set of candidates by \CC, the set of voters by \VV, the number of candidates by $m$, and the number of voters by $n$. Every voter $v_i$ has a preference or vote $\suc_i$ which is a complete order over \CC. We denote the set of complete orders over \CC by $\LL(\CC)$. We call a tuple of $n$ preferences $(\suc_1, \suc_2, \cdots, \suc_n)\in\LL(\CC)^n$ an $n$-voter preference profile. It is often convenient to view a preference as a subset of $\CC\times\CC$ --- a preference \suc corresponds to the subset $\AA = \{(x, y)\in\CC\times\CC: x\suc y\}$. For a preference \suc and a subset $\AA\subseteq\CC$ of candidates, we define $\suc(\AA)$ be the preference \suc restricted to \AA, that is $\suc(\AA) = \suc \cap (\AA\times\AA)$. Let $\uplus$ denote the disjoint union of sets. 
A map $r:\uplus_{n,|\mathcal{C}|\in\mathbb{N}^+}\mathcal{L(C)}^n\longrightarrow 2^\mathcal{C}\setminus\{\emptyset\}$
is called a \emph{voting rule}. For a voting rule $r$ and a preference profile $\succ = (\succ_1, \dots, \succ_n)$, we say a candidate $x$ wins uniquely if $r(\succ) = \{x\}$ and $x$ co-wins if $x\in r(\suc)$. For a vote $\suc\in\LL(\CC)$ and two candidates $x, y\in\CC$, we say $x$ is placed before $y$ in \suc if $x\suc y$; otherwise we say $x$ is placed after $y$ in \suc. For any two candidates $x, y\in\CC$ with $x\ne y$ in an election \EE, let us define the margin $\DD_\EE(x, y)$ of $x$ from $y$ to be $|\{ i: x \suc_i y \}| - |\{ i: y \suc_i x \}|$. Examples of some common voting rules are as follows.

{\bf Positional scoring rules:} A collection $(\overrightarrow{s_m})_{m\in\NB^+}$ of $m$-dimensional vectors $\overrightarrow{s_m}=\left(\alpha_m,\alpha_2,\dots,\alpha_1\right)\in\mathbb{N}^m$ 
 with $\alpha_m\ge\alpha_2\ge\dots\ge\alpha_1$ and $\alpha_m>\alpha_1$ for every $m\in \mathbb{N}^+$ naturally defines a voting rule --- a candidate gets score $\alpha_i$ from a vote if it is placed at the $i^{th}$ position, and the  score of a candidate is the sum of the scores it receives from all the votes. 
 The winners are the candidates with maximum score. Scoring rules remain unchanged if we multiply every $\alpha_i$ by any constant $\lambda>0$ and/or add any constant $\mu$. Hence, we assume without loss of generality that for any score vector $\overrightarrow{s_m}$, there exists a $j$ such that $\alpha_k = 0$ for all $k<j$ and the greatest common divisor of $\alpha_1, \ldots, \alpha_m$ is one. Such a $\overrightarrow{s_m}$ is called a normalized score vector. Without loss of generality, we will work with normalized scoring rules only in this work. If $\alpha_i$ is $1$ for $i\in [k]$ and $0$ otherwise, then we get the $k$-approval voting rule. For the $k$-veto voting rule, $\alpha_i$ is $0$ for $i\in [m-k]$ and $-1$ otherwise. $1$-approval is called the plurality voting rule and $1$-veto is called the veto voting rule.
 
\textbf{Copeland$^{\alpha}$:} Given $\alpha\in[0,1]$, the Copeland$^{\alpha}$ score of a candidate $x$ is $|\{y\ne x:\DD_\EE(x,y)>0\}|+\alpha|\{y\ne x:\DD_\EE(x,y)=0\}|$. The winners are the candidates with maximum Copeland$^{\alpha}$ score. If not mentioned otherwise, we will assume $\alpha$ to be zero.
 
{\bf Maximin:} The maximin score of a candidate $x$ in an election $E$ is $\min_{y\ne x} \DD_\EE(x,y)$. The winners are the candidates with maximum maximin score.

{\bf Bucklin:} Let $\ell$ be the minimum integer such that there exists at least one candidate $x\in\CC$ whom more than half of the voters place in their top $\ell$ positions. Then the Bucklin winner is the candidate who is placed most number of times within top $\el$ positions of the votes.

{\bf Elections with Incomplete Information.} A more general setting is an {\it election} where the votes are only 
\emph{partial orders} over candidates. A \emph{partial order} is a relation that is \emph{reflexive, 
antisymmetric}, and \emph{transitive}. A partial vote can be extended to possibly more than one linear votes depending on how we fix the order for the unspecified pairs of candidates.\longversion{ For example, in an election with the set of candidates $\mathcal{C} = \{a, b, c\}$, 
a valid partial vote can be $a \succ b$. This partial vote can be extended to three linear votes namely, $a \succ b \succ c$, $a \succ c \succ b$, $c \succ a \succ b$.} Given a partial vote $\suc$, we say that an extension $\suc^\pr$ of \suc {\em places the candidate $c$ as high as possible} if $a\suc^\pr c$ implies $a\suc^\prr c$ for every extension $\suc^\prr$ of \suc.

\begin{definition}($r$--\PW)\\
 Given a set of partial votes \PP over a set of candidates \CC and a candidate $c\in\CC$, does there exist an extension $\PP^\pr$ of \PP such that $c\in r(\PP^\pr)$?
\end{definition}

\longversion{We denote an arbitrary instance of \PW by $(\CC, \PP, c)$.}
\section{Results}
For ease of exposition, we present all our results for the co-winner case. All our proofs extend easily to the unique winner case too. We begin with our results for the scoring rules.
\subsection{Scoring Rules}

\nmtodo{Mention somewhere that we will be working with only normalized score vectors WLOG. ADDED THIS IN PRELIM IN THE DEFINITION OF SCORING RULES}

In this section, we establish a dichotomous result describing the status of the \PW{} problem for a large class of scoring rules when the number of undetermined pairs in every vote is at most one, two, three, or four. We begin by introducing some terminology. Instead of working directly with score vectors, it will sometimes be convenient for us to refer to the ``vector of differences'', which, for a score vector $s$ with $m$ coordinates, is a vector $d(s)$ with $m-1$ coordinates with each entry being the difference between adjacent scores corresponding to that location and the location left to it. This is formally stated below.

\begin{definition} Given a normalized score vector $\overrightarrow{s_m}=\left(\alpha_m,\alpha_{m-1},\dots,\alpha_1=0\right)\in\mathbb{N}^m$, the associated difference vector $d(\overrightarrow{s_m})$ is given by $\left(\alpha_m - \alpha_{m-1},\alpha_{m-1} - \alpha_{m-2},\dots,\alpha_2 - \alpha_1 \right)\in\mathbb{N}^{m-1}.$ We also employ the following notation to refer to the smallest score difference among all non-zero differences, and the largest score difference, respectively:
\begin{itemize}
	\item $\delta(\overrightarrow{s_m}) = \min(\{\alpha_i - \alpha_{i-1} ~|~ 2 \leq i \leq m \mbox{ and } \alpha_i - \alpha_{i-1} > 0\})$
\item $\Delta(\overrightarrow{s_m}) = \max(\{\alpha_i - \alpha_{i-1} ~|~ 2 \leq i \leq m \})$
 \end{itemize} 
\end{definition}

Note that for every normalized score vector $\overrightarrow{s_m}$, $\Delta(\overrightarrow{s_m})$ is always non-zero. We now proceed to defining the notion of \nice{} scoring rules. Consider a score vector $\overrightarrow{s_m} = \left(\alpha_m,\alpha_{m-1},\dots,\alpha_1\right)$. For $0 \leq i \leq m$, we say that $\overrightarrow{s_{m+1}}$ is obtained from $s_m$ by inserting $\alpha$ just before position $i$ from the right if: $$\overrightarrow{s_{m+1}} = \left(\alpha_m,\alpha_{m-1},\dots,\alpha_{i+1},\alpha,\alpha_{i}, \ldots, \alpha_2, \alpha_1\right).$$ Note that if $i = 0$, we have  $\overrightarrow{s_{m+1}} = \left(\alpha_m,\alpha_{m-1},\dots,\alpha_1,\alpha\right)$, and if $i = m$, then we have $\overrightarrow{s_{m+1}} = \left(\alpha,\alpha_m,\alpha_{m-1},\dots,\alpha_1\right)$. For $0 \leq i \leq m$, we say that the position $i$ is \textit{admissible} if $i = 0$, or $i = m$, or $\alpha_{i+1} = \alpha_i$.

\begin{definition}[Smooth scoring rules]
We say that a scoring rule $s$ is \nice{} if there exists some constant $n_0 \in \NB^+$ such that for all $m \geq n_0$, the score vector $\overrightarrow{s_m}$ can be obtained from $\overrightarrow{s_{m-1}}$ by inserting an additional score value at any position $i$ that is admissible. \end{definition}

Intuitively speaking, a \nice{} scoring rule is one where the score vector for $m$ candidates can be obtained by either extending the one for $(m-1)$ candidates at one of the ends, or by inserting a score between an adjacent pair of ambivalent locations (i.e, consecutive scores in the score vector with the same value). Although at a first glance it may seem that the class of \nice{} scoring rules involves an evolution from a limited set of operations, we note that all of the common scoring rules, such as plurality, veto, $k$-approval, Borda, and scoring rules of the form $(2, 1, \ldots, 1, 0)$, are \nice{}. We now turn to some definitions that will help describe the cases that appear in our classification result.

\nmtodo{Maybe we should give some examples and non-examples of \nice{} scoring rules at this point, or explain why this a reasonable class of scoring rules to consider.}


\begin{definition} Let $s=(\overrightarrow{s_m})_{m\in\NB^+}$ be a scoring rule. 
\begin{itemize}
\item We say that $s$ is a \textit{Borda-like scoring rule} if there exists some $n_0 \in \NB^+$ for which we have that $\Delta(\overrightarrow{s_m}) = \delta(\overrightarrow{s_m})$ for every $m > n_0$. 

\item Any rule that is not Borda-like is called a \textit{differentiating scoring rule}. 

\item For any vector $t$ with $\ell$ co-ordinates, we say that $s$ is $t$-difference-free if there exists some $n_0 \in \NB^+$ such that for every $m \geq n_0$, the vector $t$ does not occur in $d(\overrightarrow{s_m})$. In other words, the vector $\langle d(\overrightarrow{s_m})[i], \ldots, d(\overrightarrow{s_m})[i+\ell-1] \rangle \neq t$ for any $1 \leq i \leq m-\ell$. 

\item For any vector $t$, we say that $s$ is $t$-contaminated it is not $t$-difference-free. We also say that $s$ is $t$-contaminated at $m$ if the vector $t$ occurs in $d(\overrightarrow{s_m})$.
\end{itemize}
\end{definition}

We will frequently be dealing with Borda-like score vectors. To this end, the following easy observation will be useful. 

\begin{observation} If $s=(\overrightarrow{s_m})_{m\in\NB^+}$ is a Borda-like scoring rule in its normalized form, then there exists $n_0 \in \NB^+$ such that all the coordinates of $d(\overrightarrow{s_m})$ are either zero or one for all $m > n_0$.
\end{observation}

It turns out that if a scoring rule is smooth, then its behavior with respect to some of the properties above is fairly monotone. For instance, we have the following easy proposition. \shortversion{For the interest of space, we move proofs of some of our results including all our polynomial time algorithms to the appendix. For a few proofs, we only provide a sketch of the proof deferring the complete proof to the appendix. We mark these results with $\star$. All our polynomial time algorithms are based on reduction to the maximum flow problem in a graph.}

\begin{restatable}{proposition}{DiffNotBordaLike}\shortversion{[$\star$]}
\label{prop:diff-borda-like}
	Let $s=(\overrightarrow{s_m})_{m\in\NB^+}$ be a \nice{} scoring rule that is not Borda-like. Then there exists some $n_0 \in \NB^+$ such that $\Delta(\overrightarrow{s_m}) \neq \delta(\overrightarrow{s_m})$ for every $m > n_0$. 
\end{restatable}\longversion{

\begin{proof}
	Since $s$ is not Borda-like, there exists some $\ell \in \NB^+$ for which $\Delta(\overrightarrow{s_\ell}) \neq \delta(\overrightarrow{s_\ell})$. We claim that for any $m \geq \ell$, $\Delta(\overrightarrow{s_m}) \neq \delta(\overrightarrow{s_m})$. We prove this by induction. The base case follows directly from the assumption. Suppose the inductive hypothesis is that $\Delta(\overrightarrow{s_m}) \neq \delta(\overrightarrow{s_m})$, for some $m > \ell$, where $\overrightarrow{s_m}=\left(\alpha_m,\alpha_{m-1},\dots,\alpha_1\right)\in\mathbb{N}^m$. Since $\Delta(\overrightarrow{s_m}) \neq \delta(\overrightarrow{s_m})$, there exists $1 \leq j \leq m-1$ for which $\alpha_{j+1} - \alpha_j > 1$ (since $\Delta(\overrightarrow{s_m}) > \delta(\overrightarrow{s_m}) \ge 1$), and in particular, this implies that position $j$ is not admissible.  
	
	Now using the fact that $s$ is a \nice{} scoring rule, we let $\overrightarrow{s_{m+1}}$ be any score vector that can be obtained from $\overrightarrow{s_m}$ by inserting an additional score value at a position $i$, where we recall that $i$ must be an admissible position. Observe that $i \neq j$, so $\Delta(\overrightarrow{s_{m+1}}) \geq \Delta(\overrightarrow{s_m})$. Also inserting a score value cannot increase the smallest non-zero score difference. Therefore, $\delta(\overrightarrow{s_m}) \ge \delta(\overrightarrow{s_{m+1}})$ and the claim follows. 
\end{proof}
}

We are now ready to state the first classification result of this section, for the scenario where every vote has at most one missing pair. We use \SAT to prove some of our hardness results. The \SAT problem is the 3-SAT problem restricted to formulas in which each clause contains exactly three literals, and each variable occurs exactly twice positively and twice negatively. We know that \SAT is \NPC~\cite{ECCC-TR03-022}. Let us first present a structural result for scoring rules which we will use subsequently. 

Suppose we have a set $\CC = \{c_1, \ldots, c_{m-1}, g\}$ of $m$ candidates including a ``dummy'' candidate $g$. Then we know from \cite{baumeister2011computational,journalsDeyMN16}, that for a score vector $(\alpha_m, \ldots, \alpha_1)$ and integers $\{k_i^j\}_{i\in[m-1], j\in[m-1]}$, we can add votes polynomially many in $\sum_{i\in[m-1], j\in[m-1]} k_i^j$ so that the score of the candidate $c_i$ is $\lambda + \sum_{j\in[m-1]} k_i^j (\alpha_{j}-\alpha_{j+1})$ for some $\lambda$ and the score of $g$ is less than $\lambda$. Since the greatest common divisor of non-zero differences of the consecutive entries in a normalized score vector is one, we have the following.

\begin{lemma}\label{score_gen}
Let $\mathcal{C} = \{c_1, \ldots, c_m\} \cup D, (|D|>0)$ be a set of candidates, and $\vec{\alpha}$ a normalized score vector of length $|\mathcal{C}|$. Then for every $\mathbf{X} = (X_1, \ldots, X_m) \in \mathbb{Z}^m$, there exists $\lambda\in \mathbb{N}$ and a voting profile \VV such that the $\vec{\alpha}$-score of $c_i$ is $\lambda + X_i$ for all $1\le i\le m$,  and the score of candidates $d\in D$ is less than $\lambda$. Moreover, the number of votes in \VV is $O(poly(|\mathcal{C}|\cdot \sum_{i=1}^m |X_i|))$.
\end{lemma}

\begin{restatable}{theorem}{OneMissingPair}\shortversion{[$\star$]}
\label{thm:one-missing-pair}
Let $s$ be a \nice{} scoring rule. If $s$ is differentiating, then the \PW problem is \NPC{}, even if every vote has at most one undetermined pair of candidates. Otherwise, the \PW problem for $s$ is in \Pb if every vote has at most one undetermined pair of candidates.
\end{restatable}

\nmtodo{I have not changed the proof below so that it is easy to merge your changes. Please change the start of the proof to "For the hardness result, we reduce from..." and add the maxflow argument at the end. Also, make a reference to Proposition 1 above.} 
  
\begin{proof} For the hardness result, we reduce from an instance of $(3,B2)$-SAT. Let $\II$ be an instance of $(3,B2)$-SAT, over the variables $\VV = \{x_1, \ldots, x_n\}$ and with clauses $\TT = \{c_1, \ldots, c_t\}$. To construct the reduced instance $\II^\pr$, we introduce two candidates for every variable, and one candidate for every clause, one special candidate $w$, and a dummy candidate $g$ to achieve desirable score differences. Notationally, we will use $b_i$ (corresponding to $x_i$) and $b_i^\prime$ (corresponding to $\bar{x}_i$) to refer to the candidates based on the variable $x_i$ and $e_j$ to refer to the candidate based on the clause $c_j$. To recap, the set of candidates are given by:
$$\CC = \{ b_i, b_i^\pr ~|~ x_i \in \VV\} \cup \{e_j ~|~ c_j \in \TT \} \cup \{w, g\}.$$

Consider an arbitrary but fixed ordering over $\CC$, such as the lexicographic order. In this proof, the notation $\overrightarrow{\CC^\pr}$ for any $\CC^\pr \subseteq \CC$ will be used to denote the lexicographic ordering restricted to the subset $\CC^\pr$. Let $m$ denote $|\CC| = 2n + t + 2$, and let $\overrightarrow{s_m}=\left(\alpha_m,\alpha_{m-1},\dots,\alpha_1\right)\in\mathbb{N}^m$. Since $s$ is a \nice{} differentiating scoring rule, we have that there exist $1 \leq p,q \leq m$ such that $|p-q| > 1$ and $ \alpha_p - \alpha_{p-1} > \alpha_q - \alpha_{q-1} \geq 1.$

We use $D$ to refer to the larger of the two differences above, namely $\alpha_p - \alpha_{p-1}$ and $d$ to refer to $\alpha_q - \alpha_{q-1}$. We now turn to a description of the votes. Fix an arbitrary subset $\CC_1$ of $(m-p)$ candidates. For every variable $x_i \in \VV$, we introduce the following complete and partial votes.
$$\ppp_i := \overrightarrow{\CC_1} \succ b_i \succ b_i^\pr \succ \overrightarrow{\CC \setminus \CC_1} \mbox{ and } \ppp_i^\pr := \ppp_i \setminus \{(b_i,b_i^\pr)\}$$

We next fix an arbitrary subset $\CC_2\subset\CC$ of $(m-q)$ candidates. Consider a literal $\ell$ corresponding to the variable $x_i$. We use $\ell^\star$ to refer to the candidate $b_j$ if the literal is positive and $b_j^\pr$ if the literal is negated. For every clause $c_j \in \TT$ given by $c_j = \{\ell_1, \ell_2, \ell_3\}$, we introduce the following complete and partial votes.
$$\qqq_{j,1} := \overrightarrow{\CC_2} \succ e_j \succ \ell_1^\star \succ \overrightarrow{\CC \setminus \CC_2} \mbox{ and } \qqq_{j,1}^\pr := \qqq_{j,1} \setminus \{(e_j,\ell_1^\star)\}$$
$$\qqq_{j,2} := \overrightarrow{\CC_2} \succ e_j \succ \ell_2^\star \succ \overrightarrow{\CC \setminus \CC_2} \mbox{ and } \qqq_{j,2}^\pr := \qqq_{j,2} \setminus \{(e_j,\ell_2^\star)\}$$
$$\qqq_{j,3} := \overrightarrow{\CC_2} \succ e_j \succ \ell_3^\star \succ \overrightarrow{\CC \setminus \CC_2} \mbox{ and } \qqq_{j,3}^\pr := \qqq_{j,3} \setminus \{(e_j,\ell_3^\star)\}$$

Let us define the following sets of votes:
$$\PP = \left( \bigcup_{i=1}^n \ppp_i \right) \cup \left(\bigcup_{\substack{1 \leq j \leq t,}\\\substack{1\leq b \leq 3}} \qqq_{j,b} \right)
\mbox{ and } \PP^\pr = \left( \bigcup_{i=1}^n \ppp_i^\pr \right) \cup \left(\bigcup_{\substack{1 \leq j \leq t,}\\ \substack{1\leq b \leq 3}} \qqq_{j,b}^\pr \right) $$ 

There exists a set of complete votes \WW of size polynomial in $m$ with the following properties due to \Cref{score_gen}. Let $s^+: \CC \longrightarrow\NB$ be a function mapping candidates to their scores from the set of votes $\PP \cup\WW$. Then $\WW$ can be constructed to ensure the scores as in \Cref{tbl:score_differentiating}. We now define the instance $\II^\pr$ of \PW to be $(\CC, \PP^\pr \cup \WW, w)$. This completes the description of the reduction. We now turn to a proof of the equivalence. Before we begin making our arguments, observe that since $w$ does not participate in any undetermined pairs of the votes in $\PP^\pr$, it follows that the score of $w$ continues to be $s^+(w)$ in any completion of $\PP^\pr$. The intuition for the construction, described informally, is as follows. The score of every ``clause candidate'' needs to decrease by $d$, which can be achieved by pushing it down against its literal partner in the $\qqq_j$-votes. However, this comes at the cost of increasing the score of the literals by $2d$ (since every literal appears in at most two clauses). It turns out that this can be compensated appropriately by ensuring that the candidate corresponding to the literal appears in the $(p-1)^{th}$ position among the $\ppp$-votes, which will adjust for this increase. Therefore, the setting of the $(b_i^\pr,b_i)$ pairs in a successful completion of $\ppp_i$ can be read off as a signal for how the corresponding variable should be set by a satisfying assignment. \shortversion{We defer the formal proof of equivalence of the two instances and the polynomial time solvable case to the appendix.}
%
\begin{table}
\centering
 \begin{tabular}{|c|c|}\hline
  $s^+(e_j) = s^+(w) + d~\forall~1 \leq j \leq t$ & $s^+(b_i) = s^+(w) + 1 - d~\forall~1 \leq i \leq n$\\
  $s^+(b_i^\pr) = s^+(w) + 1 - d - D~\forall~1 \leq i \leq n$ & $s^+(g) < s^+(w)$ \\\hline
 \end{tabular}
 \caption{Score of candidates from $\PP \cup\WW$.}\label{tbl:score_differentiating}
\end{table}
\end{proof}

We make a couple of quick remarks before moving on to our next result. Observe that any hardness result that holds for instances where every vote has at most $k$ undetermined pairs also holds for instances where every vote has at most $k^\prime$ undetermined pairs with $k^\prime > k$, by a standard special case argument. Therefore, the next question for us to address is that of whether the \PW problem is in \Pb for all Borda-like scoring rules when the number of undetermined pairs in every vote is at most two.\longversion{ It is not hard to see that the maxflow algorithm that we used previously does not immediately work. For example, if we have a scenario with the partial vote $ \{\cdots x \succ y \succ z \succ \cdots \} \setminus \{(x,y),(x,z)\},$ then the maxflow construction in the previous proof, if used directly, may lead us to a solution that completes the vote in accordance with $x \succ y$ and $z \succ x$, violating transitivity. In fact, we will use this kind of a structure to engineer a reduction from a problem called \TDM{}, which is known to be \NPC, and is defined as follows.}\shortversion{ We show that the complexity of the \PW problem for the Borda-like scoring rules crucially depends on the presence (or absence) some particular patterns in the score vector. We begin with a hardness result which uses a reduction from the \TDM problem. The \TDM problem is known to be \NPC and is defined as follows.} 

\begin{definition}[\TDM]
 Given three disjoint sets $\XX, \YY,$ and \ZZ of size $t$ each and a collection \SS of subsets of $\XX\cup\YY\cup\ZZ$ each containing exactly one element from \XX, \YY, and \ZZ, does there exist a sub-collection $\SS^\pr\subset\SS$ of size $t$ such that $\cup_{\AA\in\SS^\pr} \AA = \XX \cup \YY \cup \ZZ$.
\end{definition}

To help us deal with the nature of the score vectors considered, we will use the following proposition, which again reflects the monotonicity property alluded to earlier.  

\begin{restatable}{proposition}{OneOneContaminated}\shortversion{[$\star$]}
\label{prop:11-contaminated}
	Let $s$ be a normalized \nice{} scoring rule that is not $\langle 1, 1 \rangle$-difference-free. Then there exists some $n_0 \in \NB^+$ such that for every $m \geq n_0$, $s$ is  $\langle 1, 1 \rangle$-contaminated at $m$.
\end{restatable}\longversion{

\begin{proof}
If $s$ is not $\langle 1, 1 \rangle$-difference-free, then there exists some $\ell \in \NB^+$ for which $s$ is $\langle 1, 1 \rangle$-contaminated at $\ell$. In particular, this implies that there exists an index $i$ for which $\alpha_{i+1} - \alpha_i = 1$ and $\alpha_i - \alpha_{i-1} = 1$. We now argue that $s$ is $\langle 1, 1 \rangle$-contaminated at $m$ for every $m \geq \ell$. This follows from the fact that the differences $(\alpha_{i+1} - \alpha_i)$ and $(\alpha_i - \alpha_{i-1})$  are ``carried forward''. In particular since the positions $i-1$ and $i$ are not admissible, it is not possible to diminish these differences in any score vector $s_{\ell+1}$ obtained from $s_\ell$, and repeating this argument for all $m \geq \ell$ gives us the desired claim. 
\end{proof}
}


We are now ready to state our next result, which shows that if there are at most $2$ undetermined pairs of candidates in every vote, and we are dealing with a smooth Borda-like scoring rule $s$, then the \PW{} problem is \NPC{} if $s$ is $\langle 1,1 \rangle$-contaminated, and solvable in polynomial time otherwise.

\begin{restatable}{theorem}{TwoMissingPairs}\shortversion{[$\star$]}
\label{thm:two-missing-pairs}
Let $s$ be a \nice{}, Borda-like scoring rule. If $s$ is $\langle 1,1 \rangle$-contaminated, the \PW problem is \NPC{}, even if every vote has at most $2$ undetermined pairs of candidates.  On the other hand, if $s$ is $\langle 1,1 \rangle$-difference-free, then the \PW problem for $s$ is in \Pb if every vote has at most $2$ undetermined pairs of candidates.
\end{restatable}

\nmtodo{Again, please change the start of the proof to "For the hardness result, we reduce from..." and add the maxflow argument at the end, just indicating that it is similar to the previous maxflow argument. Also, make a reference to Proposition 2 above, and use the language of the difference vector to start off the proof.} 
We now address the case involving at most $3$ undetermined pairs in every vote. The interesting scoring rules here are smooth Borda-like scoring rules that are $\langle 1, 1 \rangle$-difference-free. It turns out that here, if the scoring rule is further $\langle 1, 0, 1 \rangle$-difference-free, then the problem again admits a maxflow formulation. On the other hand,  $s$ is $\langle 1,0,1 \rangle$-contaminated at $m\ge N_0$ for osme constant $N_0$, then the \PW problem is \NPC even with $3$ undetermined pairs of candidates per vote. 

\begin{restatable}{theorem}{ThreeMissingPairs}\shortversion{[$\star$]}
\label{thm:three-missing-pairs}
Let $s$ be a \nice{}, Borda-like, $\langle 1,1 \rangle$-difference-free scoring rule. If there exists a constant $N_0 \in \NB^+$ such that $s$ is $\langle 1,0,1 \rangle$-contaminated for all $m \geq n_0$, then the \PW problem is \NPC{}, even if every vote has at most $3$ undetermined pairs.  On the other hand, if $s$ is $\langle 1,0,1 \rangle$-difference-free, then the \PW problem for $s$ is in \Pb if every vote has at most $3$ undetermined pairs.
\end{restatable}
	\nmtodo{This is the 3DM construction and the easy-but-with-many-cases maxflow construction. I plan to add the hardness proof a little later, maybe you can work on the construction of the maxflow instance?}

{\bf Remark.} Note that unlike the previous two results, this statement is not a complete classification, because we don't have an appropriate analog of Propositions~\ref{prop:diff-borda-like} and~\ref{prop:11-contaminated}. Having said that, our result holds for a more general class of scoring rules: those where $s$ is $\langle 1,0,1 \rangle$-contaminated at $m$ ``sufficiently'' often, that is to say that if $\overrightarrow{s_m}$ is $\langle 1,0,1 \rangle$-contaminated and $m^\pr > m$ is the smallest natural number for which $\overrightarrow{s_m}$ is $\langle 1,0,1 \rangle$-contaminated, then $m^\pr - m$ is bounded by some polynomial function of $m$, by inserting appropriately many dummy candidates using standard techniques. 	

We now turn to our final result for scoring rules. Let $s$ be a smooth, Borda-like scoring rule that is $\langle 1,1 \rangle$-difference-free. Then we have the following. If $s$ is $\langle 0,1,0 \rangle$-contaminated, then the \PW problem for $s$ is \NPC{} even when every vote has at most $4$ undetermined pairs of candidates. If $s$ is $\langle 0,1,0 \rangle$-difference-free, then notice that $d(\overrightarrow{s_m})$ for any suitably large $m \in \NB^+$ can contain at most two ones (since $s$ is also $\langle 1,1 \rangle$-difference-free). If the number of ones in $d(\overrightarrow{s_m})$ is one, then $d(\overrightarrow{s_m})$ either has a one on the first or the last coordinate (recall that $s$ is $\langle 0,1,0 \rangle$-difference-free), corresponding to the plurality and veto voting rules, respectively. On the other hand, if the number of ones is two, $d(\overrightarrow{s_m}) = \langle 1,0, \ldots, 0, 1\rangle$, which is equivalent (in normal form) to the scoring rule $(2,1,\ldots,1,0)$. The \PW problem is polynomial time solvable for plurality and veto voting rules, and we show here that it is also polynomially solvable for the scoring rule  $(2,1,\ldots,1,0)$ as long as the number of undetermined pairs of candidates in any vote is at most $m-1$. If we allow for $m$ or more undetermined pairs of candidates in every vote, then we show that the \PW problem is \NPC{}. As before, we will need the following property of $\langle 1, 0, 1 \rangle$-contaminated vectors. 

\begin{restatable}{proposition}{ZeroOneZeroContaminated}\shortversion{[$\star$]}
\label{prop:010-contaminated}
	Let $s$ be a normalized \nice{} scoring rule that is not $\langle 1, 0, 1 \rangle$-difference-free. Then there exists some $n_0 \in \NB^+$ such that $s$ is $\langle 0, 1, 0 \rangle$-contaminated at $m$ for every $m > n_0$. 
\end{restatable}\longversion{

\begin{proof}
If $s$ is not $\langle 0, 1, 0 \rangle$-difference-free, then there exists some $\ell \in \NB^+$ for which $s$ is $\langle 0, 1, 0 \rangle$-contaminated at $\ell$. In particular, this implies that the score vector admits the pattern $(\alpha, \alpha, \alpha+1, \alpha+1)$. Let the positions (counted from the bottom) for these scores be $i$, $i-1$, $i-2$ and $i-3$, respectively. Now note that $i-2$ is not an admissible position, and it follows that any score vector $s_{\ell+1}$ obtained from $s_\ell$ will therefore continue to be $\langle 0, 1, 0 \rangle$-contaminated. Repeating this argument for all $m \geq \ell$ gives us the desired claim. 
\end{proof}
}

We now state the final result in this section. It is easily checked that the result accounts for all smooth, Borda-like scoring rules that are  $\langle 1,1 \rangle$-difference-free. 

\begin{restatable}{theorem}{FourMissingPairs}\shortversion{[$\star$]}
\label{thm:four-missing-pairs}
Let $s$ be a smooth, Borda-like scoring rule that is $\langle 1,1 \rangle$-difference-free. Then we have the following. 

\begin{enumerate}
\item If $s$ is $\langle 0,1,0 \rangle$-contaminated, then the \PW problem for $s$ is \NPC{} even when every vote has at most $4$ undetermined pairs of candidates. 
\item If $s$ is equivalent to $(2,1,\ldots,1,0)$, then \PW is \NPC{} even when the number of undetermined pairs of candidates in every vote is at most $m-1$. 
\item If $s$ is equivalent to $(2,1,\ldots,1,0)$ and the number of undetermined pairs of candidates is strictly less than $m-1$, then \PW is in \Pb.
\item If $s$ is neither $\langle 0,1,0 \rangle$-contaminated nor equivalent to $(2,1,\ldots,1,0)$, then $s$ is equivalent to either the plurality or veto scoring rules and \PW is in \Pb for these cases. 
\end{enumerate}
\end{restatable}

\subsection{Copeland$^\alpha$ Voting Rule}

We now turn to the Copeland$^\alpha$ voting rule. We show in \Cref{thm:copeland_hard_two} below that the \PW problem is \NPC for the Copeland$^\alpha$ voting rule even when every vote has at most $2$ undetermined pairs of candidates for every $\alpha\in[0,1]$.

\begin{restatable}{theorem}{CopelandTwo}\shortversion{[$\star$]}
\label{thm:copeland_hard_two}
 The \PW problem is \NPC for the Copeland$^\alpha$ voting rule even if the number of undetermined pairs of candidates in every vote is at most $2$ for every $\alpha\in[0,1]$.
\end{restatable}

We prove in \Cref{thm:copeland_poly_one} that the number of undetermined pairs of candidates in \Cref{thm:copeland_hard_two} is tight for the Copeland$^0$ and Copeland$^1$ voting rules.

\begin{restatable}{theorem}{CopelandPoly}\shortversion{[$\star$]}
\label{thm:copeland_poly_one}
 The \PW problem is in \Pb for the Copeland$^0$ and Copeland$^1$ voting rules if the number of undetermined pairs of candidates in every vote is at most $1$.
\end{restatable}

We show next that the \PW problem is \NPC for the Copeland$^\alpha$ voting rule even if the number of undetermined pairs of candidates in every vote is at most $1$ for $\alpha\in(0,1)$. We break the proof into two parts --- \Cref{lem:copeland_hard_alpha_zero_half} proves the result for every $\alpha\in(0,\nfrac{1}{2}]$ and \Cref{lem:copeland_hard_alpha_half_one} proves for every $\alpha\in[\nfrac{1}{2}, 1)$.

\begin{restatable}{lemma}{CopelandZeroHalf}\shortversion{[$\star$]}
\label{lem:copeland_hard_alpha_zero_half}
 The \PW problem is \NPC for the Copeland$^\alpha$ voting rule even if the number of undetermined pairs in every vote is at most $1$ for every $\alpha\in(0,\nfrac{1}{2}]$.
\end{restatable}

\begin{proof}
 The \PW problem for the Copeland$^\alpha$ voting rule is clearly in \NP. To prove \NP-hardness of \PW, we reduce \PW from \SAT. Let $\II$ be an instance of \SAT, over the variables $\VV = \{x_1, \ldots, x_n\}$ and with clauses $\TT = \{c_1, \ldots, c_m\}$. We construct an instance $\II^\pr$ of \PW from \II as follows. 
 \[\text{Set of candidates: } \CC = \{x_i, \bar{x}_i, d_i: i\in[n]\}\cup\{c_i:i\in[m]\}\cup\{c\}\cup\GG, \text{ where } \GG = \{g_1, \ldots, g_{mn}\} \]
 
 For every $i\in[n]$, let us consider the following votes $\ppp_{x_i}^1, \ppp_{x_i}^2, \ppp_{\bar{x}_i}^1, \ppp_{\bar{x}_i}^2$. 
 \[ \ppp_{x_i}^1, \ppp_{x_i}^2: x_i\suc d_i\suc \text{others}~,~ \ppp_{\bar{x}_i}^1, \ppp_{\bar{x}_i}^2: \bar{x}_i\suc d_i\suc \text{others} \]
 
 Using $\ppp_{x_i}^1, \ppp_{x_i}^2, \ppp_{\bar{x}_i}^1, \ppp_{\bar{x}_i}^2$, we define the partial votes $\ppp_{x_i}^{1\pr}, \ppp_{x_i}^{2\pr}, \ppp_{\bar{x}_i}^{1\pr}, \ppp_{\bar{x}_i}^{2\pr}$ as follows. 
 \[ \ppp_{x_i}^{1\pr}, \ppp_{x_i}^{2\pr}: \ppp_{x_i}^1 \setminus \{(x_i, d_i)\}~,~ \ppp_{\bar{x}_i}^{1\pr}, \ppp_{\bar{x}_i}^{2\pr}: \ppp_{\bar{x}_i}^1 \setminus \{(\bar{x}_i, d_i)\}\]
 
 Let a clause $c_j$ involves the literals $\el_j^1, \el_j^2, \el_j^3$. For every $j\in[m]$, let us consider the following votes $\qqq_j({\el_j^1}), \qqq_j({\el_j^2}), \qqq_j({\el_j^3})$. 
 \[ \qqq_j({\el_j^k}): c_j\suc \el_j^k\suc \text{others}, \forall k\in[3] \]
 
 Using $\qqq_j({\el_j^1}), \qqq_j({\el_j^2}), \qqq_j({\el_j^3})$, we define the partial votes $\qqq_j^\pr({\el_j^1}), \qqq_j^\pr({\el_j^2}), \qqq_j^\pr({\el_j^3})$ as follows. 
 \[ \qqq_j^\pr({\el_j^k}): \qqq_j({\el_j^k})\setminus\{(c_j, \el_j^k)\}, \forall k\in[3] \]
 
 Let us define $$\PP = \cup_{i\in[n]} \{\ppp_{x_i}^1, \ppp_{x_i}^2, \ppp_{\bar{x}_i}^1, \ppp_{\bar{x}_i}^2\} \cup_{j\in[m]} \{\qqq_j({\el_j^1}), \qqq_j({\el_j^2}), \qqq_j({\el_j^3})\}$$ and $$\PP^\pr = \cup_{i\in[n]} \{\ppp_{x_i}^{1\pr}, \ppp_{x_i}^{2\pr}, \ppp_{\bar{x}_i}^{1\pr}, \ppp_{\bar{x}_i}^{2\pr}\} \cup_{j\in[m]} \{\qqq_j^\pr({\el_j^1}), \qqq_j^\pr({\el_j^2}), \qqq_j^\pr({\el_j^3})\}.$$ \shortversion{There exists a set of complete votes \QQ of size polynomial in $n$ and $m$ which realizes \Cref{tbl:cop_initial_alpha_zero_half}~\cite{mcgarvey1953theorem}. All the wins and defeats in \Cref{tbl:cop_initial_alpha_zero_half} are by a margin of $2$.} We now define the instance $\II^\pr$ of \PW to be $(\CC, \PP^\pr \cup \QQ, c)$. Notice that the number of undetermined pairs of candidates in every vote in $\II^\pr$ is at most $1$. This finishes the description of the \PW instance. \shortversion{We defer the formal proof of equivalence of the two instances and the polynomial time solvable case to the appendix.}\longversion{We claim that \II and $\II^\pr$ are equivalent.}
 \begin{table}[!htbp]
  \centering
  \resizebox{\textwidth}{!}{  
  \begin{tabular}{|ccccc|}\hline
   Candidates & Copeland$^\alpha$ score & Winning against & Losing against & Tie with\\\hline\hline
   
   $c$ & \makecell{$(2n+m)\alpha$\\$ + n + \nfrac{3mn}{4}$} & $G^\pr\subset\GG, |G^\pr| = n + \nfrac{3mn}{4}$ & \makecell{$\GG\setminus G^\pr,|G^\pr|=n + \nfrac{3mn}{4}$ \\$ d_i, \forall i\in[n]$} & \makecell{$x_i, \bar{x}_i\forall i\in[n]$ \\ $c_j \forall j\in[m]$}\\
   
   $x_i, \forall i\in[n]$ & \makecell{$(2n+m)\alpha$\\$ + n + \nfrac{3mn}{4}$} & \makecell{$G^\prr\subset\GG, |G^\prr| = \nfrac{3mn}{4}$ \\$d_i \forall i\in[n]$} & $\GG\setminus (G^\pr\cup G^\prr)$ & \makecell{$c, G^\pr\subset\GG, |G^\pr| = m $ \\$x_j,   
   \forall j\in[n]\setminus\{i\}$\\$\bar{x}_j \forall j\in[n]$} \\
   
   $\bar{x}_i, \forall i\in[n]$ & \makecell{$(2n+m)\alpha$\\$ + n + \nfrac{3mn}{4}$} & \makecell{$G^\prr\subset\GG, |G^\prr| = \nfrac{3mn}{4}$ \\$d_i \forall i\in[n]$} & $\GG\setminus (G^\pr\cup G^\prr)$ & \makecell{$c, G^\pr\subset\GG, |G^\pr| = m $ \\$\bar{x}_j, \forall j\in[n]\setminus\{i\}$\\$x_j \forall j\in[n]$} \\
   
   $c_j, \forall j\in[m]$ & \makecell{$(2n+m-1)\alpha $\\$+ n + \nfrac{3mn}{4}+1$} & \makecell{$x_i, \bar{x}_i\forall i\in[n]$\\$G^\pr\subset\GG, |G^\pr| = \nfrac{3mn}{4}-n+1$} & \makecell{$\GG\setminus (G^\pr\cup G^\prr)$ \\$ d_i, \forall i\in[n]$} & \makecell{$c$\\$c_j \forall j\in[m]\setminus\{i\}$\\$G^\prr\subset\GG, |G^\prr| = 2n-1$} \\
   
   $d_i, i\in[n]$& \makecell{$(2n+m)\alpha$\\$ + n + \nfrac{3mn}{4}-1$} &\makecell{$c, c_j, \forall j\in[m]$\\$G^\prr\subset\GG, |G^\prr| = \nfrac{3mn}{4} - m + n -2$} & \makecell{$x_i, \bar{x}_i\forall i\in[n]$\\$\GG\setminus (G^\pr\cup G^\prr)$} & $G^\pr\subset\GG, |G^\pr| = 2n+m$ \\
   
   $g_i, \forall i\in[mn]$ & $ < \nfrac{3mn}{4}$ &  & $\forall j\in\{i+k: k\in [\lfloor\nfrac{(mn-1)}{2}\rfloor]$ &  \\\hline
  \end{tabular}}
  \caption{Summary of Copeland$^\alpha$ scores of the candidates from $\PP\cup\QQ$. All the wins and defeats in the table are by a margin of $2$.}\label{tbl:cop_initial_alpha_zero_half}
 \end{table}
\end{proof}

Next we present \Cref{lem:copeland_hard_alpha_half_one} which resolves the complexity of the \PW problem for the Copeland$^\alpha$ voting rule for every $\alpha\in[\nfrac{1}{2},1)$ when every partial vote has at most one undetermined pair of candidates.

\begin{restatable}{lemma}{CopelandHalfOne}\shortversion{[$\star$]}
\label{lem:copeland_hard_alpha_half_one}
 The \PW problem is \NPC for the Copeland$^\alpha$ voting rule even if the number of undetermined pairs in every vote is at most $1$ for every $\alpha\in[\nfrac{1}{2},1)$.
\end{restatable}

We get the following result for the Copeland$^\alpha$ voting rule from \Cref{lem:copeland_hard_alpha_zero_half,lem:copeland_hard_alpha_half_one}.

\begin{theorem}\label{thm:copeland_hard_one}
 The \PW problem is \NPC for the Copeland$^\alpha$ voting rule even if the number of undetermined pairs of candidates in every vote is at most $1$ for every $\alpha\in(0,1)$.
\end{theorem}

\subsection{Maximin and Bucklin Voting Rules}

%

To prove our hardness result for the maximin voting rule, we reduce the \PW problem from the $d$--\MIS problem which is defined as below. $d$--\MIS is known to be \NPC (for example, see this \cite{DBLP:books/sp/CyganFKLMPPS15}). We denote arbitrary instance of $d$--\MIS by $\left(\VV = \cup_{i=1}^k \VV_k, \EE \right)$.

\begin{definition}[$d$--\MIS]
 Given a $d$-regular graph $\GG = (\VV, \EE)$, an integer $k$, and a partition of the set of vertices \VV into $k$ independent sets $\VV_1, \ldots, \VV_k$, that is $\VV = \cup_{i\in[k]}\VV_i$ and $\VV_i$ is an independent set for every $i\in[k]$, does there exists an independent set $\SS\subset \VV$ in \GG such that $|\SS\cap\VV_i|=1$ for every $i\in[k]$.
\end{definition}

Now we prove our hardness result for the \PW problem for the maximin voting rule in \Cref{thm:maximin_hard}.

\begin{theorem}\label{thm:maximin_hard}
 The \PW problem is \NPC for the maximin voting rule even if the number of undetermined pairs of candidates in every vote is at most $2$.
\end{theorem}

\begin{proof}
 The \PW problem for the maximin voting rule is clearly in \NP. To prove \NP-hardness of \PW, we reduce \PW from $d$--\MIS. Let $\II = \left(\VV = \cup_{i=1}^k \VV_k, \EE \right)$ be an arbitrary instance of $d$--\MIS. We construct an instance $\II^\pr$ of \PW from \II as follows. 
 \[\text{Set of candidates: } \CC = \VV\cup\EE\cup\{c\}\cup\{g_i, g_i^\pr: i\in[k]\} \]
 
 For every $u\in\VV_i$ and $\el\in[d]$, let us consider the following vote $\ppp_u$. 
 \[ \ppp_u^\el = \overrightarrow{(\CC\setminus\{u, g_i, g_i^\pr\})_u} \suc g_i \suc g_i^\pr \suc u, \text{where } \overrightarrow{(\CC\setminus\{u, g_i, g_i^\pr\})_u} \text{ is any fixed ordering of } \CC\setminus\{u, g_i, g_i^\pr\} \]
 
 Using $\ppp_u^\el$, we define a partial vote $\ppp_u^{\pr\el}$ as follows. 
 \[ \ppp_u^{\pr\el} = \ppp_u^\el \setminus \{(g_i, u), (g_i^\pr, u)\} \]
 
 For every edge $e=(u_i, u_j)$ where $u_i\in\VV_i$ and $u_j\in\VV_j$, let us consider the following votes $\ppp_{e,u_i}$ and $\ppp_{e,u_j}$. 
 \[ \ppp_{e,u_i} = \overrightarrow{(\CC\setminus\{u_i, g_i^\pr, e\})} \suc e \suc g_i^\pr \suc u_i~,~ \ppp_{e,u_j} = \overrightarrow{(\CC\setminus\{u_j, g_j^\pr, e\})} \suc e \suc g_j^\pr \suc u_j \]
 
 Using $\ppp_{e,u_i}$ and $\ppp_{e,u_j}$, we define the partial votes $\ppp_{e,u_i}^\pr$ and $\ppp_{e,u_j}^\pr$ as follows. 
 \[ \ppp_{e,u_i}^\pr = \ppp_{e,u_i} \setminus \{(e,u_i), (g_i^\pr, u_i)\}~,~ \ppp_{e,u_j}^\pr = \ppp_{e,u_j} \setminus \{(e,u_j), (g_j^\pr, u_j)\} \]
 
 Let us call $\ppp_e = \{\ppp_{e,u_i}, \ppp_{e,u_j}\}$ and $\ppp_e^\pr = \{\ppp_{e,u_i}^\pr, \ppp_{e,u_j}^\pr\}$. Let us define $\PP = \cup_{u\in\VV,\el\in[d]} \ppp_u^\el \cup_{e\in\EE} \ppp_e$ and $\PP^\pr = \cup_{u\in\VV,\el\in[d]} \ppp_u^{\pr\el} \cup_{e\in\EE} \ppp_e^\pr$. There exists a set of complete votes \QQ of size polynomial in $|\VV|$ and $|\EE|$ with the pairwise margins as in \Cref{tbl:maximin_extra_votes}~\cite{mcgarvey1953theorem}. Let $\lambda > 3d$ be any positive even integer.
%
%
 \begin{table}[!htbp]
 \centering
  \begin{tabular}{|cc|}\hline
   $\forall e\in\EE, \DD_{\PP\cup\QQ}(e,c) = \lambda$ & $\forall i\in[k], \forall u\in\VV_i, \DD_{\PP\cup\QQ} (u, g_i) = \lambda - 2d $\\
   $\forall i\in[k], \forall u\in\VV_i, \DD_{\PP\cup\QQ} (g_i^\pr, u) = \lambda + 2d$ & $\forall i\in[k], e\in\EE, \DD_{\PP\cup\QQ} (e, g_i^\pr) = \lambda$ \\
   \multicolumn{2}{|c|}{$\forall e=(u_i, u_j)\in\EE, \DD_{\PP\cup\QQ} (u_i, e) = \DD_{\PP\cup\QQ} (u_j, e) = \lambda-2$}\\\hline
  \end{tabular}
  \caption{Pairwise margins of candidates from $\PP\cup\QQ$.}\label{tbl:maximin_extra_votes}
 \end{table}

 For every pair of candidates $(c_i, c_j)\in\CC\times\CC$ whose pairwise margin is not defined above, we define $\DD_{\PP\cup\QQ}(c_i, c_j) = 0$. We summarize the maximin score of every candidate in $\PP \cup \QQ$ in \Cref{tbl:maximin_initial}. We now define the instance $\II^\pr$ of \PW to be $(\CC, \PP^\pr \cup \QQ, c)$. Notice that the number of undetermined pairs of candidates in every vote in $\II^\pr$ is at most $2$. This finishes the description of the \PW instance. We claim that \II and $\II^\pr$ are equivalent. 
 \begin{table}[!htbp]
  \centering
  \begin{tabular}{|ccc||ccc|}\hline
   Candidates & maximin score & Worst against & Candidates & maximin score & Worst against\\\hline\hline
   $c$ & $-\lambda$ & $e\in\EE$ & $(u_i, u_j)\in\EE$ & $-(\lambda-2)$ & $u_i, u_j$\\
   $u\in\VV_i$ & $-(\lambda+2d)$ & $g_i^\pr$ & $g_i$ & $-(\lambda-2d)$ & $u\in\VV_i$\\
    $g_i^\pr$ & $-\lambda$ & $e\in\EE$&&&\\\hline
  \end{tabular}
  \caption{Summary of initial Copeland scores of the candidates}\label{tbl:maximin_initial}
 \end{table}
 
 In the forward direction, suppose that \II be a \YES instance of $d$--\MIS. Then there exists $u_i\in\VV_i$ for every $i\in[k]$ such that $\UU=\{u_i: i\in[k]\}$ forms an independent set. We extend the partial vote $\ppp_u^{\pr\el}$ for every $u\in\VV_i, i\in[k], \el\in[d]$ to $\bar{\ppp}_u^\el$ as follows. 
 $$
 \bar{\ppp}_u^\el = 
 \begin{cases}
  \overrightarrow{(\CC\setminus\{u, g_i, g_i^\pr\})_u} \suc u \suc g_i \suc g_i^\pr & u\in\UU\\
  \overrightarrow{(\CC\setminus\{u, g_i, g_i^\pr\})_u} \suc g_i \suc g_i^\pr \suc u & u\notin\UU
 \end{cases}
 $$
 
 For every $e=(u_i, u_j)$, we extend $\ppp_{e,u_i}^\pr$ and $\ppp_{e,u_j}^\pr$ to $\bar{\ppp}_{e,u_i}$ and $\bar{\ppp}_{e,u_j}$. Since \UU is an independent set, at least one of $u_i$ and $u_j$ does not belong to \UU. Without loss of generality, let us assume $u_i\notin \UU$. 
 \[ \bar{\ppp}_{e,u_i} = \overrightarrow{(\CC\setminus\{u_i, g_i^\pr, e\})} \suc u_i \suc e \suc g_i^\pr~,~ \bar{\ppp}_{e,u_j} = \overrightarrow{(\CC\setminus\{u_j, g_j^\pr, e\})} \suc e \suc g_j^\pr \suc u_j\]
 
 Let us call $\bar{\ppp}_e = \{\bar{\ppp}_{e,u_i}, \bar{\ppp}_{e,u_j}\}$. We consider the extension of \PP to $\bar{\PP} = \cup_{u\in\VV,\el\in[d]} \bar{\ppp}_u^\el \cup_{e\in\EE} \bar{\ppp}_e$.  We claim that $c$ is a co-winner in the profile $\bar{\PP}\cup\QQ$ since the maximin score of $c$, $g_i, g_i^\pr$ for every $i\in[k]$, $u\in\VV$, and $e\in\EE$ in $\bar{\PP}\cup\QQ$ is $-\lambda$.
 
 In the reverse direction suppose the \PW instance $\II^\pr$ be a \YES instance. Then there exists an extension of the set of partial votes $\PP^\pr$ to a set of complete votes $\bar{\PP}$ such that $c$ is a co-winner in $\bar{\PP}\cup\QQ$. Let us call the extension of $\ppp_u^{\pr\el}$ in $\bar{\PP}$ $\bar{\ppp}_u^\el$, $\ppp_{e,u_i}^\pr$ and $\ppp_{e,u_j}^\pr$ in $\bar{\PP}$ $\bar{\ppp}_{e,u_i}$ and $\bar{\ppp}_{e,u_j}$ respectively. First we notice that the maximin score of $c$ in $\bar{\PP}\cup\QQ$ is $-\lambda$ since the relative ordering of $c$ with respect to every other candidate is already fixed in $\PP^\pr\cup\QQ$. Now we observe that, in $\PP\cup\QQ$, the maximin score of $g_i$ for every $i\in[k]$ is $-(\lambda-2d)$. Hence, for $c$ to co-win, there must exists at least one $u_i^*\in\VV_i$ for every $i\in[k]$ such that $u_i^*\suc g_i\suc g_i^\pr$ in $\bar{\ppp}_{u_i^*}^\el$ for every $\el\in[d]$. We claim that $\UU = \{u_i^* : i\in[k]\}$ is an independent set in \II. If not, then suppose there exists an edge $e$ between $u_i^*$ and $u_j^*$ for some $i,j\in[k]$. Now notice that, for $c$ to co-win either $u_i^*\suc e\suc g_i^\pr$ in $\bar{\ppp}_{e,u_i^*}$ or $u_j^*\suc e\suc g_j^\pr$ in $\bar{\ppp}_{e,u_j^*}$. However, this makes the maximin score of either $u_i^*$ or $u_j^*$ strictly more than $-\lambda$ contradicting our assumption that $c$ co-wins the election. Hence, \UU forms an independent set in \II.
\end{proof}

We next prove in \Cref{thm:maximin_poly_one} that the maximum number of undetermined pairs of candidates in \Cref{thm:maximin_hard} is tight.

\begin{restatable}{theorem}{MaximinPoly}\shortversion{[$\star$]}
\label{thm:maximin_poly_one}
 The \PW problem is in \Pb for the maximin voting rule if the number of undetermined pairs of candidates in every vote is at most $1$.
\end{restatable}

Finally, we state our results for the Bucklin voting rule.

\begin{restatable}{theorem}{BucklinPoly}\shortversion{[$\star$]}
\label{thm:bucklin_hard} The \PW problem is \NPC for the Bucklin voting rule even if the number of undetermined pairs of candidates in every vote is at most $2$, and is in \Pb if the number of undetermined pairs of candidates in every vote is at most $1$.
\end{restatable}

\section{Conclusion}

We have demonstrated the exact minimum number of undetermined pairs allowed per vote which keeps the \PW winner problem \NPC, and we were able to address a large class of scoring rules, Copeland$^\alpha$, maximin, and Bucklin voting rules. Our results generalize many of the known hardness results in the literature, and show that for many voting rules, we need a surprisingly small number of undetermined pairs (often just one or two) for the \PW problem to be \NPC{}. In the context of scoring rules, it would be interesting to extend these tight results to the class of pure scoring rules, and to extend Theorem~\ref{thm:three-missing-pairs} to account for all smooth scoring rules. 




\bibliography{winner}


\newpage

\section*{Appendix}

\DiffNotBordaLike*

\begin{proof}
	Since $s$ is not Borda-like, there exists some $\ell \in \NB^+$ for which $\Delta(\overrightarrow{s_\ell}) \neq \delta(\overrightarrow{s_\ell})$. We claim that for any $m \geq \ell$, $\Delta(\overrightarrow{s_m}) \neq \delta(\overrightarrow{s_m})$. We prove this by induction. The base case follows directly from the assumption. Suppose the inductive hypothesis is that $\Delta(\overrightarrow{s_m}) \neq \delta(\overrightarrow{s_m})$, for some $m > \ell$, where $\overrightarrow{s_m}=\left(\alpha_m,\alpha_{m-1},\dots,\alpha_1\right)\in\mathbb{N}^m$. Since $\Delta(\overrightarrow{s_m}) \neq \delta(\overrightarrow{s_m})$, there exists $1 \leq j \leq m-1$ for which $\alpha_{j+1} - \alpha_j > 1$ (since $\Delta(\overrightarrow{s_m}) > \delta(\overrightarrow{s_m}) \ge 1$), and in particular, this implies that position $j$ is not admissible.  
	
	Now using the fact that $s$ is a \nice{} scoring rule, we let $\overrightarrow{s_{m+1}}$ be any score vector that can be obtained from $\overrightarrow{s_m}$ by inserting an additional score value at a position $i$, where we recall that $i$ must be an admissible position. Observe that $i \neq j$, so $\Delta(\overrightarrow{s_{m+1}}) \geq \Delta(\overrightarrow{s_m})$. Also inserting a score value cannot increase the smallest non-zero score difference. Therefore, $\delta(\overrightarrow{s_m}) \ge \delta(\overrightarrow{s_{m+1}})$ and the claim follows. 
\end{proof}

\OneMissingPair*

\begin{proof}
 For the hardness result, we reduce from an instance of $(3,B2)$-SAT. Let $\II$ be an instance of $(3,B2)$-SAT, over the variables $\VV = \{x_1, \ldots, x_n\}$ and with clauses $\TT = \{c_1, \ldots, c_t\}$. 

To construct the reduced instance $\II^\pr$, we introduce two candidates for every variable, and one candidate for every clause, one special candidate $w$, and a dummy candidate $g$ to achieve desirable score differences. Notationally, we will use $b_i$ (corresponding to $x_i$) and $b_i^\prime$ (corresponding to $\bar{x}_i$) to refer to the candidates based on the variable $x_i$ and $e_j$ to refer to the candidate based on the clause $c_j$. To recap, the set of candidates are given by:

$$\CC = \{ b_i, b_i^\pr ~|~ x_i \in \VV\} \cup \{e_j ~|~ c_j \in \TT \} \cup \{w, g\}.$$

Consider an arbitrary but fixed ordering over $\CC$, such as the lexicographic order. In this proof, the notation $\overrightarrow{\CC^\pr}$ for any $\CC^\pr \subseteq \CC$ will be used to denote the lexicographic ordering restricted to the subset $\CC^\pr$. Let $m$ denote $|\CC| = 2n + t + 2$, and let $\overrightarrow{s_m}=\left(\alpha_m,\alpha_{m-1},\dots,\alpha_1\right)\in\mathbb{N}^m$. Since $s$ is a \nice{} differentiating scoring rule, we have that there exist $1 \leq p,q \leq m$ such that $|p-q| > 1$ and the following holds:

$$ \alpha_p - \alpha_{p-1} > \alpha_q - \alpha_{q-1} \geq 1 $$

We use $D$ to refer to the larger of the two differences above, namely $\alpha_p - \alpha_{p-1}$ and $d$ to refer to $\alpha_q - \alpha_{q-1}$. We now turn to a description of the votes. Fix an arbitrary subset $\CC_1$ of $(m-p)$ candidates. For every variable $x_i \in \VV$, we introduce the following complete and partial votes.

$$\ppp_i := \overrightarrow{\CC_1} \succ b_i \succ b_i^\pr \succ \overrightarrow{\CC \setminus \CC_1} \mbox{ and } \ppp_i^\pr := \ppp_i \setminus \{(b_i,b_i^\pr)\}$$

We next fix an arbitrary subset $\CC_2\subset\CC$ of $(m-q)$ candidates. Consider a literal $\ell$ corresponding to the variable $x_i$. We use $\ell^\star$ to refer to the candidate $b_j$ if the literal is positive and $b_j^\pr$ if the literal is negated. For every clause $c_j \in \TT$ given by $c_j = \{\ell_1, \ell_2, \ell_3\}$, we introduce the following complete and partial votes.

$$\qqq_{j,1} := \overrightarrow{\CC_2} \succ e_j \succ \ell_1^\star \succ \overrightarrow{\CC \setminus \CC_2} \mbox{ and } \qqq_{j,1}^\pr := \qqq_{j,1} \setminus \{(e_j,\ell_1^\star)\}$$

$$\qqq_{j,2} := \overrightarrow{\CC_2} \succ e_j \succ \ell_2^\star \succ \overrightarrow{\CC \setminus \CC_2} \mbox{ and } \qqq_{j,2}^\pr := \qqq_{j,2} \setminus \{(e_j,\ell_2^\star)\}$$

$$\qqq_{j,3} := \overrightarrow{\CC_2} \succ e_j \succ \ell_3^\star \succ \overrightarrow{\CC \setminus \CC_2} \mbox{ and } \qqq_{j,3}^\pr := \qqq_{j,3} \setminus \{(e_j,\ell_3^\star)\}$$

Let us define the following sets of votes:

$$\PP = \left( \bigcup_{i=1}^n \ppp_i \right) \cup \left(\bigcup_{\substack{1 \leq j \leq t,}\\\substack{1\leq b \leq 3}} \qqq_{j,b} \right)$$

and 

$$\PP^\pr = \left( \bigcup_{i=1}^n \ppp_i^\pr \right) \cup \left(\bigcup_{\substack{1 \leq j \leq t,}\\ \substack{1\leq b \leq 3}} \qqq_{j,b}^\pr \right) $$ 

There exists a set of complete votes \WW of size polynomial in $m$ with the following properties due to \Cref{score_gen}. Let $s^+: \CC \longrightarrow\NB$ be a function mapping candidates to their scores from the set of votes $\PP \cup\WW$. Then $\WW$ can be constructed to ensure that the following hold.

\begin{itemize}
	\item $s^+(e_j) = s^+(w) + d$ for all $1 \leq j \leq t$. 
	\item $s^+(b_i) = s^+(w) + 1 - d$ for all $1 \leq i \leq n$. 
	\item $s^+(b_i^\pr) = s^+(w) + 1 - d - D$ for all $1 \leq i \leq n$. 
	\item $s^+(g) < s^+(w)$
\end{itemize}

We now define the instance $\II^\pr$ of \PW to be $(\CC, \PP^\pr \cup \WW, w)$. This completes the description of the reduction. We now turn to a proof of the equivalence. Before we begin making our arguments, observe that since $w$ does not participate in any undetermined pairs of the votes in $\PP^\pr$, it follows that the score of $w$ continues to be $s^+(w)$ in any completion of $\PP^\pr$. The intuition for the construction, described informally, is as follows. The score of every ``clause candidate'' needs to decrease by $d$, which can be achieved by pushing it down against its literal partner in the $\qqq_j$-votes. However, this comes at the cost of increasing the score of the literals by $2d$ (since every literal appears in at most two clauses). It turns out that this can be compensated appropriately by ensuring that the candidate corresponding to the literal appears in the $(p-1)^{th}$ position among the $\ppp$-votes, which will adjust for this increase. Therefore, the setting of the $(b_i^\pr,b_i)$ pairs in a successful completion of $\ppp_i$ can be read off as a signal for how the corresponding variable should be set by a satisfying assignment.

We now turn to a formal proof. In the forward direction, let $\tau: \VV \rightarrow \{0,1\}$ be a satisfying assignment for $\II$. Then we have the following completions of the votes in $\PP^\pr$. To begin with, for all $1\leq i \leq n$, we have: 

\begin{equation*}
  \ppp_i^{\pr\pr} := \left\{
  \begin{array}{rl}
\overrightarrow{\CC_1} \succ b_i^\pr \succ b_i \succ\overrightarrow{\CC \setminus \CC_1} & \text{if } \tau(x_i) = 1,\\ 
\overrightarrow{\CC_1} \succ b_i \succ b_i^\pr \succ \overrightarrow{\CC \setminus \CC_1} & \text{if } \tau(x_i) = 0.
\end{array} \right.	
\end{equation*}

For a clause $c_j = \{\ell_1, \ell_2, \ell_3\}$, suppose $\tau(\ell_1) = 1$. Then we have the following completions for the votes $\qqq_{j,b}$, $1 \leq b \leq 3$:

$$\qqq^{\pr\pr}_{j,1} := \overrightarrow{\CC_2} \succ \ell_1^\star \succ e_j \succ \overrightarrow{\CC \setminus \CC_2}, $$

$$\qqq^{\pr\pr}_{j,2} := \overrightarrow{\CC_2} \succ e_j \succ \ell_2^\star \succ \overrightarrow{\CC \setminus \CC_2}, $$

$$\qqq^{\pr\pr}_{j,3} := \overrightarrow{\CC_2} \succ e_j \succ \ell_3^\star \succ \overrightarrow{\CC \setminus \CC_2}$$

The completions for the cases when $\tau(\ell_2) = 1$ or $\tau(\ell_3) = 1$ are analogously defined. Now consider the election given by the complete votes described above, which we denote by $\PP^{\pr\pr}$. Let $s^\star: \CC \longrightarrow\NB$ be the function that maps candidates to their scores from the votes $\PP^{\pr\pr} \cup\WW$. Then, we have the following. 

\begin{itemize}
	\item Since $\tau$ is a satisfying assignment, for every $1 \leq j \leq t$, we have that the candidate $e_j$ swaps places with one of its companions in at least one of the votes $\qqq_{j,b}$, $1 \leq b \leq 3$. Therefore, it loses a score of at least $d$, leading to the observation that $s^\pr(e_j) \leq s^+(e_j) - d  = s^+(w)$ for all $1 \leq j \leq t$. 
	\item We now turn to a candidate $b_i$, for some $1 \leq i \leq n$. If $\tau(x_i) = 0$, then notice that the score of $b_i$ does not change, and therefore $s^\star(b_i) = s^+(b_i) = s^+(w) - d + 1 \leq s^+(w)$, since $d \geq 1$. Otherwise, note that it decreases by $D$ and increases by at most $2d$, implying that $s^\star(b_i) = s^+(b_i) + 2d - D =  s^+(w) + 1 - D + d \leq s^+(w)$, as $(D-d) \geq 1$. 
	\item Finally, consider the candidates $b_i^\pr$, for for some $1 \leq i \leq n$. If $\tau(x_i) = 1$, then notice that the score of $b_i^\pr$ increases by $D$, and therefore $s^\star(b_i^\pr) = s^+(b_i^\pr) + D = s^+(w) - d + 1 \leq s^+(w)$, since $d \geq 1$. Otherwise, note that its score increases by at most $2d$, implying that $s^\star(b_i^\pr) = s^+(b_i^\pr) + 2d =  s^+(w) + 1 - D + d \leq s^+(w)$, as $(D-d) \geq 1$. 
\end{itemize}

This completes the forward direction of the argument. In the other direction, let $\PP^{\pr\pr}$ be any completion of the votes in $\PP^\pr$ which makes $w$ a co-winner with respect to $s$. Let $s^\star$ be the function that computes the scores of all the candidates with respect to $\PP^{\pr\pr}$. We define the following assignment to the variables of $\II$ based on $\PP^{\pr\pr}$:

\begin{equation*}
  \tau(x_i) := \left\{
  \begin{array}{rl}
1 & \text{if } \overrightarrow{\CC_1} \succ b_i^\pr \succ b_i \succ\overrightarrow{\CC \setminus \CC_1} \in \PP^{\pr\pr},
\\ 
0 & \text{if }\overrightarrow{\CC_1} \succ b_i \succ b_i^\pr \succ \overrightarrow{\CC \setminus \CC_1} \in \PP^{\pr\pr}.
\end{array} \right.	
\end{equation*}

We claim that $\tau$, as defined above, satisfies every clause in $\II$. Consider any clause $c_j \in \TT$. Observe that the score of the corresponding candidate, $e_j$, must decrease by at least $d$ in any valid completion, since $s^+(e_j) = s^+(w) + d$. Therefore, for at least one of the votes $\qqq_{j,b}$, $1 \leq b \leq 3$, we must have a completion where $e_j$ appears at position $q-1$. We claim that the literal $\ell$ that consequently appears at position $q$ must be set to one by $\tau$. Indeed, suppose not. Then we have two cases, as follows:

\begin{itemize}
	\item Suppose the literal $\ell$ corresponds to the positive appearance of a variable $x_j$. If $\tau(x_j) = 0$, then the score of $b_j$ has increased by $d$, making its final score equal to $s^+(w) + 1$, which is a contradiction.
	\item Suppose the literal $\ell$ corresponds to the negated appearance of a variable $x_j$. If $\tau(x_j) = 1$, then the score of $b_j$ has increased by $D + d$, making its final score equal to $s^+(w) + 1$, which is, again, a contradiction.
\end{itemize}

Now we turn to the proof of the polynomial time solvable case. Let the input instance of \PW be $(\CC, \PP, c)$ where every partial vote in \PP has at most one pair of candidates whose ordering is undetermined. In every partial vote in \PP we place the candidate $c$ as high as possible. Suppose in a partial vote \ppp in \PP, one undetermined pair of candidates appears at positions $i$ and $i+1$ (from the bottom) and $\alpha_i = \alpha_{i+1}$. Then we fix the ordering of the undetermined pair of candidates in \ppp arbitrarily. Let us call the resulting profile $\PP^\pr$. It is easy to see that $(\CC, \PP, c)$ is a \YES instance if and only if $(\CC, \PP^\pr, c)$ is a \YES instance. Notice that the position of $c$ in every vote in $\PP^\pr$ is fixed and thus we know the score of $c$; let it be $s(c)$. Also we can compute the minimum score that every candidate receives over all extensions of $\PP^\pr$. Let $s(w)$ be the minimum score of candidate $w$. If there exists a candidate $z$ such that $s(z)>s(c)$, then we output \NO. Otherwise we construct the following flow graph $\GG=(\VV, \EE)$. For every partial vote \vvv in $\PP^\pr$, we add a vertex $v_\vvv$ in \VV. We also add a vertex $v_w$ in \VV for every candidate $w$ other than $c$. We also add two special vertices $s$ and $t$ in \VV. We add an edge from $s$ to $v_\vvv$ for every $\vvv\in\PP^\pr$ of capacity $1$, an edge from $v_w$ to $t$ of capacity $s(c)-s(w)$ for every candidate $w$ other than $c$. If a vote $\vvv\in\PP^\pr$ has an undetermined pair $(x, y)$ of candidates, we add an edge from $v_\vvv$ to $v_x$ and $v_y$ each of capacity $1$. Let the number of votes in $\PP^\pr$ which are not complete be $t$. Now it is easy to see that the $(\CC, \PP^\pr, c)$ is a \YES instance if and only if there is a flow of size $t$ in \GG.
\end{proof}

\OneOneContaminated*

\begin{proof}
If $s$ is not $\langle 1, 1 \rangle$-difference-free, then there exists some $\ell \in \NB^+$ for which $s$ is $\langle 1, 1 \rangle$-contaminated at $\ell$. In particular, this implies that there exists an index $i$ for which $\alpha_{i+1} - \alpha_i = 1$ and $\alpha_i - \alpha_{i-1} = 1$. We now argue that $s$ is $\langle 1, 1 \rangle$-contaminated at $m$ for every $m \geq \ell$. This follows from the fact that the differences $(\alpha_{i+1} - \alpha_i)$ and $(\alpha_i - \alpha_{i-1})$  are ``carried forward''. In particular since the positions $i-1$ and $i$ are not admissible, it is not possible to diminish these differences in any score vector $s_{\ell+1}$ obtained from $s_\ell$, and repeating this argument for all $m \geq \ell$ gives us the desired claim.
\end{proof}

\TwoMissingPairs*

\begin{proof}
 Since the scoring rule is $\langle 1,1 \rangle$-contaminated, for every $\el\ge N_0$ for some constant $N_0$, there exists an index $i\in[\el-2]$ in the score vector $(\alpha_j)_{j\in[\el]}$ such that $\alpha_{i+2} - \alpha_{i+1} = \alpha_{i+1} - \alpha_i = 1$. We begin with the proof of hardness. The \PW problem is clearly in \NP. To prove \NP-hardness of \PW, we reduce \PW from \TDM. Let $\II = (\XX\cup\YY\cup\ZZ, \SS)$ be an arbitrary instance of \TDM. Let $|\XX|=|\YY|=|\ZZ|= m > N_0$. We construct an instance $\II^\pr$ of \PW from \II as follows. 
 \[ \CC = \XX\cup\YY\cup\ZZ\cup\{c,d\} \]
 
 For every $\sss=(x, y, z)\in\SS$, let us consider the following vote $\ppp_\sss$. 
 \[ \ppp_\sss = \overrightarrow{(\CC\setminus\CC_\sss)} \suc x \suc y \suc z \suc \overrightarrow{\CC_\sss}, \text{ for some fixed } \CC_\sss\subset(\CC\setminus\{x, y, z\}) \text{ with } |\CC_\sss| = i-1 \]
 
 Using $\ppp_\sss$, we define a partial vote $\ppp_\sss^\pr$ as follows. 
 \[ \ppp_\sss^\pr = \ppp_\sss \setminus \{(x,y), (x,z)\} \]
 
 Let us define $\PP = \cup_{\sss\in\SS} \ppp_\sss$ and $\PP^\pr = \cup_{\sss\in\SS} \ppp_\sss^\pr$. There exists a set of complete votes \QQ of size polynomial in $m$ with the scores as in \Cref{tbl:score_one_one_contaminated_appendix} due to \Cref{score_gen}. Let $s_{\PP\cup\QQ}:\CC\longrightarrow\NB$ be a function mapping candidates to their scores from the set of votes $\PP\cup\QQ$.
 
 
 \begin{table}[!htbp]
 \centering
  \begin{tabular}{|c|c|}\hline
   $s_{\PP\cup\QQ} (x) = s_{\PP\cup\QQ} (c) + 2, ~\forall x\in\XX$ & $s_{\PP\cup\QQ} (y) = s_{\PP\cup\QQ} (c) - 1, \forall y\in\YY$\\
   $s_{\PP\cup\QQ} (z) = s_{\PP\cup\QQ} (c) - 1, \forall z\in\ZZ$ & $s_{\PP\cup\QQ} (d) < s_{\PP\cup\QQ} (c)$\\\hline
  \end{tabular}
  \caption{Score of candidates from $\PP \cup\WW$.}\label{tbl:score_one_one_contaminated_appendix}
 \end{table}

 We now define the instance $\II^\pr$ of \PW to be $(\CC, \PP^\pr \cup \QQ, c)$. Notice that the number of undetermined pairs in every vote in $\II^\pr$ is at most $2$. This finishes the description of the \PW instance. We claim that \II and $\II^\pr$ are equivalent. 
 
 In the forward direction, suppose that \II be a \YES instance of \TDM. Then, there exists a collection of $m$ sets $\SS^\pr\subset\SS$ in \SS such that $\cup_{\AA\in\SS^\pr} \AA = \XX\cup\YY\cup\ZZ$. We extend the partial vote $\ppp_\sss^\pr$ to $\bar{\ppp}_\sss$ as follows for $\sss\in\SS$. 
 $$
 \bar{\ppp}_\sss = 
 \begin{cases}
  \overrightarrow{(\CC\setminus\CC_\sss)} \suc y \suc z \suc x \suc \overrightarrow{\CC_\sss} & \sss\in\SS^\pr\\
  \overrightarrow{(\CC\setminus\CC_\sss)} \suc x \suc y \suc z \suc \overrightarrow{\CC_\sss} & \sss\notin\SS^\pr
 \end{cases}
 $$
 
 We consider the extension of \PP to $\bar{\PP} = \cup_{\sss\in\SS} \bar{\ppp}_\sss$. We claim that $c$ is a co-winner in the profile $\bar{\PP}\cup\QQ$ since $s_{\bar{\PP}\cup\QQ} (c) = s_{\bar{\PP}\cup\QQ} (x) = s_{\bar{\PP}\cup\QQ} (y) = s_{\bar{\PP}\cup\QQ} (z) > s_{\bar{\PP}\cup\QQ} (d)$.
 
 For the reverse direction, suppose the \PW instance $\II^\pr$ be a \YES instance. Then there exists an extension of the set of partial votes $\PP^\pr$ to a set of complete votes $\bar{\PP}$ such that, $c$ is a co-winner in $\bar{\PP}\cup\QQ$. Let us call the extension of $\ppp_\sss^\pr$ in $\bar{\PP}$ $\bar{\ppp}_\sss$. We first claim that, for every $x\in\XX$, there exists exactly one $\sss\in\SS$ such that $\bar{\ppp}_\sss = \overrightarrow{(\CC\setminus\CC_\sss)} \suc y \suc z \suc x \suc \overrightarrow{\CC_\sss}$. Notice that, the score of $c$ is same in every extension of $\PP^\pr$. Hence, for $c$ to co-win, every candidate $x\in\XX$ must lose $\alpha_{i+2} - \alpha_i$ points. If there are more than one vote in $\bar{\PP}$ where $x$ is placed after some candidate $y\in\YY$, then the total increase of scores of all the candidates in \YY is more than $m(\alpha_{i+2} - \alpha_{i+1})$ and thus there exists a candidate $y^\pr\in\YY$ whose score has increased by strictly more than $\alpha_{i+2} - \alpha_{i+1}$. However, in such a scenario, the score of $y^\pr$ will be strictly more than the score of $c$ contradicting the fact that $c$ is a co-winner in $\bar{\PP}\cup\QQ$. Now, the claim follows from the observation that, every $x\in\XX$ must lose $\alpha_{i+2} - \alpha_i$ scores in order to $c$ co-win. Let $\SS^\pr\subseteq\SS$ be the collections of sets $\sss\in\SS$ such that $x\in\sss$ is placed after $z\in\sss$ in $\bar{\ppp}_\sss$. From the claim above, we have $|\SS^\pr|=m$. We now claim that, $\cup_\sss\in\SS^\pr = \XX\cup\YY\cup\ZZ$. Indeed, otherwise there exists a candidate $a\in\YY\cup\ZZ$ who does not belong to $\cup_\sss\in\SS^\pr$. But then the score of $a$ is strictly more than the score of $c$ contradicting the fact that $c$ is a co-winner in $\bar{\PP}\cup\QQ$. Hence, $\II^\pr$ is also a \YES instance
 
 The proof for the polynomial time solvable case is similar to the polynomial time solvable case in \Cref{thm:one-missing-pair}.
\end{proof}

\ThreeMissingPairs*

\begin{proof}
 For every $\el\ge N_0$, there exists an index $i\in[\el-2]$ in the score vector $(\alpha_j)_{j\in[\el]}$ such that $\alpha_{i+3} - \alpha_{i+2} = \alpha_{i+1} - \alpha_i = 1$ and $\alpha_{i+2} = \alpha_{i+1}$. Let $\alpha_i = \alpha$. We begin with the proof of hardness. The \PW problem is clearly in \NP. To prove \NP-hardness of \PW, we reduce \PW from \TDM. Let $\II = (\XX\cup\YY\cup\ZZ, \SS)$ be an arbitrary instance of \TDM. Let $|\XX|=|\YY|=|\ZZ|= m > N_0$. We construct an instance $\II^\pr$ of \PW from \II as follows.
 
 \[\text{Set of candidates: } \CC = \XX\cup\YY\cup\ZZ\cup\{c,d\} \]
 
 For every $\sss=(x, y, z)\in\SS$, let us consider the following vote $\ppp_\sss$.
 
 \[ \ppp_\sss = \overrightarrow{(\CC\setminus\CC_\sss)} \suc x \suc y \suc d \suc z \suc \overrightarrow{\CC_\sss}, \text{ for some fixed } \CC_\sss\subset(\CC\setminus\{x, y, z\}) \text{ with } |\CC_\sss| = i-1 \]
 
 Using $\ppp_\sss$, we define a partial vote $\ppp_\sss^\pr$ as follows.
 
 \[ \ppp_\sss^\pr = \ppp_\sss \setminus \{(x,y), (x,d), (x,z)\} \]
 
 Let us define $\PP = \cup_{\sss\in\SS} \ppp_\sss$ and $\PP^\pr = \cup_{\sss\in\SS} \ppp_\sss^\pr$. There exists a set of complete votes \QQ of size polynomial in $m$ with the following properties due to \Cref{score_gen}. Let $s_{\PP\cup\QQ}:\CC\longrightarrow\NB$ be a function mapping candidates to their scores from the set of votes $\PP\cup\QQ$.
 
 \begin{itemize}
  \item $s_{\PP\cup\QQ} (x) = s_{\PP\cup\QQ} (c) + 2, ~\forall x\in\XX$
  \item $s_{\PP\cup\QQ} (y) = s_{\PP\cup\QQ} (c) -1, \forall y\in\YY$
  \item $s_{\PP\cup\QQ} (z) = s_{\PP\cup\QQ} (c) -1, \forall z\in\ZZ$
  \item $s_{\PP\cup\QQ} (d) < s_{\PP\cup\QQ} (c)$
 \end{itemize}
 
 We now define the instance $\II^\pr$ of \PW to be $(\CC, \PP^\pr \cup \QQ, c)$. Notice that the number of undetermined pairs of candidates in every vote in $\II^\pr$ is at most $3$. This finishes the description of the \PW instance. We claim that \II and $\II^\pr$ are equivalent.
 
 In the forward direction, suppose that \II be a \YES instance of \TDM. Then there exists a collection of $m$ sets $\SS^\pr\subset\SS$ in \SS such that $\cup_{\AA\in\SS^\pr} \AA = \XX\cup\YY\cup\ZZ$. We extend the partial vote $\ppp_\sss^\pr$ to $\bar{\ppp}_\sss$ as follows for $\sss\in\SS$.
 
 $$
 \bar{\ppp}_\sss = 
 \begin{cases}
  \overrightarrow{(\CC\setminus\CC_\sss)} \suc y \suc d \suc z \suc x \suc \overrightarrow{\CC_\sss} & \sss\in\SS^\pr\\
  \overrightarrow{(\CC\setminus\CC_\sss)} \suc x \suc y \suc d \suc z \suc \overrightarrow{\CC_\sss} & \sss\notin\SS^\pr
 \end{cases}
 $$
 
 We consider the extension of \PP to $\bar{\PP} = \cup_{\sss\in\SS} \bar{\ppp}_\sss$. We claim that $c$ is a co-winner in the profile $\bar{\PP}\cup\QQ$ since $s_{\bar{\PP}\cup\QQ} (c) = s_{\bar{\PP}\cup\QQ} (x) = s_{\bar{\PP}\cup\QQ} (y) = s_{\bar{\PP}\cup\QQ} (z) > s_{\bar{\PP}\cup\QQ} (d)$.
 
 For the reverse direction, suppose the \PW instance $\II^\pr$ be a \YES instance. Then there exists an extension of the set of partial votes $\PP^\pr$ to a set of complete votes $\bar{\PP}$ such that $c$ is a co-winner in $\bar{\PP}\cup\QQ$. Let us call the extension of $\ppp_\sss^\pr$ in $\bar{\PP}$ $\bar{\ppp}_\sss$. We first claim that, for every $x\in\XX$, there exists exactly one $\sss\in\SS$ such that $\bar{\ppp}_\sss = \overrightarrow{(\CC\setminus\CC_\sss)} \suc y \suc d \suc z \suc x \suc \overrightarrow{\CC_\sss}$. Notice that, the score of $c$ is same in every extension of $\PP^\pr$. Hence, for $c$ to co-win, every candidate $x\in\XX$ must lose $2$ points. If there are more than one vote in $\bar{\PP}$ where $x$ is placed after some candidate $y\in\YY$, then the total increase of scores of all the candidates in \YY is more than $m$ and thus there exists a candidate $y^\pr\in\YY$ whose score has increased by strictly more than $2$. However, in such a scenario, the score of $y^\pr$ will be strictly more than the score of $c$ contradicting the fact that $c$ is a co-winner in $\bar{\PP}\cup\QQ$. Now the claim follows from the observation that, every $x\in\XX$ must lose $2$ scores in order to $c$ co-win. Let $\SS^\pr\subseteq\SS$ be the collections of sets $\sss\in\SS$ such that $x\in\sss$ is placed after $z\in\sss$ in $\bar{\ppp}_\sss$. From the claim above, we have $|\SS^\pr|=m$. We now claim that, $\cup_\sss\in\SS^\pr = \XX\cup\YY\cup\ZZ$. Indeed, otherwise there exists a candidate $a\in\YY\cup\ZZ$ who does not belong to $\cup_\sss\in\SS^\pr$. But then the score of $a$ is strictly more than the score of $c$ contradicting the fact that $c$ is a co-winner in $\bar{\PP}\cup\QQ$. Hence, $\II^\pr$ is also a \YES instance.
 
 We now turn to the polynomial time solvable case. Let the input instance of \PW be $(\CC, \PP, c)$ where every partial vote in \PP has at most $3$ pairs of candidates whose ordering is undetermined. In every partial vote in \PP we place the candidate $c$ as high as possible. Suppose in a partial vote \ppp in \PP, one undetermined pair of candidates appears at positions $i$ and $i+1$ (from the bottom) and $\alpha_i = \alpha_{i+1}$. Then we fix the ordering of the undetermined pair of candidates in \ppp arbitrarily. Let us call the resulting profile $\PP^\pr$. It is easy to see that $(\CC, \PP, c)$ is a \YES instance if and only if $(\CC, \PP^\pr, c)$ is a \YES instance. Notice that the position of $c$ in every vote in $\PP^\pr$ is fixed and thus we know the score of $c$; let it be $s(c)$. Also we can compute the minimum score that every candidate receives over all extensions of $\PP^\pr$. Let $s(w)$ be the minimum score of candidate $w$. If there exists a candidate $z$ such that $s(z)>s(c)$, then we output \NO. Otherwise we construct the following flow graph $\GG=(\VV, \EE)$. For every partial vote \vvv in $\PP^\pr$, we add a vertex $v_\vvv$ in \VV. We also add a vertex $v_w$ in \VV for every candidate $w$ other than $c$. We also add two special vertices $s$ and $t$ in \VV. We add an edge from $v_w$ to $t$ of capacity $s(c)-s(w)$ for every candidate $w$ other than $c$. Consider a partial $\vvv\in\PP^\pr$ where the three undetermined pairs of candidates be $(x_1, x_2), (y_1, y_2), (z_1, z_2)$. We add edges from $v_\vvv$ to $v_w$ for candidate $w$ other than $c$ as follows.
 
 \begin{itemize}
  \item If the sets $\{x_1, x_2\}, \{y_1, y_2\}$, and $\{z_1, z_2\}$ are mutually disjoint, then we add three vertices $v_\vvv(x_1, x_2), v_\vvv(y_1, y_2),$ and $v_\vvv(z_1, z_2)$, add edges from $v_\vvv$ to each of them each of capacity $1$, add edges from $v_\vvv(x_1, x_2)$ to $v_{x_1}$ and $v_{x_2}$, edges from $v_\vvv(y_1, y_2)$ to $v_{y_1}$ and $v_{y_2}$, edges from $v_\vvv(z_1, z_2)$ to $v_{z_1}$ and $v_{z_2}$ each with capacity $1$ and an edge from $s$ to $v_\vvv$ for every $\vvv\in\PP^\pr$ of capacity $3$.
  
  \item If $\{x_1, x_2\}$ and $\{y_1, y_2\}$ are each disjoint with $\{z_1, z_2\}$ and $\{x_1, x_2\}\cap\{y_1, y_2\}=\{x_1\}=\{y_1\}$, then, without loss of generality, let us assume $x_2\ge y_2$ in \vvv. Now observe that since the scoring rule is $\langle 1,0,1 \rangle$-difference-free, exactly one of $x_2$ and $y_2$ gets a score of one in every extension of \vvv; say $x_2$ gets a score of one in every extension of \vvv. Also observe that exactly one of $x_1$ and $y_2$ gets a score of $1$ in any extension of \vvv. Hence, we add an edge from $v_\vvv$ to $x_1$ and another edge from $v_\vvv$ to $y_2$ each with capacity $1$.
  
  \item If $|\{x_1, x_2\}\cup\{y_1, y_2\}\cup\{z_1, z_2\}|=3$ (say $\{x_1, x_2\}\cup\{y_1, y_2\}\cup\{z_1, z_2\} = \{a_1, a_2, a_3\}$), then either exactly one of $a_i, i\in[3]$ gets a score of $1$ in every extension of \vvv or exactly two of $a_i, i\in[3]$ gets a score of $1$ in every extension of \vvv. We add an edge from $s$ to $v_\vvv$ for every $\vvv\in\PP^\pr$ of capacity $1$ in the former case and of capacity $2$ in the later case.
 \end{itemize}
 Now it is easy to see that the $(\CC, \PP^\pr, c)$ is a \YES instance if and only if there is a flow of size $t$ in \GG.
\end{proof}

\ZeroOneZeroContaminated*

\begin{proof}
If $s$ is not $\langle 0, 1, 0 \rangle$-difference-free, then there exists some $\ell \in \NB^+$ for which $s$ is $\langle 0, 1, 0 \rangle$-contaminated at $\ell$. In particular, this implies that the score vector admits the pattern $(\alpha, \alpha, \alpha+1, \alpha+1)$. Let the positions (counted from the bottom) for these scores be $i$, $i-1$, $i-2$ and $i-3$, respectively. Now note that $i-2$ is not an admissible position, and it follows that any score vector $s_{\ell+1}$ obtained from $s_\ell$ will therefore continue to be $\langle 0, 1, 0 \rangle$-contaminated. Repeating this argument for all $m \geq \ell$ gives us the desired claim. 
\end{proof}

\FourMissingPairs*

\begin{proof} The proof of this theorem is described in four parts corresponding to the four statements above.

 {\bf Proof of Part 1.} A careful reading of the proof of Theorem 2 in \cite{XiaC11} reveals that the \PW winner problem is \NPC even when every vote has at most $4$ undetermined pairs for any scoring rule for which there exists an index $i$ such that $\alpha_{i+3} = \alpha_{i+2} = \alpha_{i+1} + 1 = \alpha_{i} + 1$. Hence, our result follows immediately.

 {\bf Proof of Part 2.} The reduction is similar in spirit to the construction used in the proof of Theorem~\ref{thm:one-missing-pair}. We describe it in detail for the sake of completeness. As before, we reduce from an instance of $(3,B2)$-SAT. Let $\II$ be an instance of $(3,B2)$-SAT, over the variables $\VV = \{x_1, \ldots, x_n\}$ and with clauses $\TT = \{c_1, \ldots, c_t\}$. 

To construct the reduced instance $\II^\pr$, we introduce four candidates for every variable, and one candidate for every clause, one special candidate $w$, and a dummy candidate $g$ to achieve desirable score differences. Notationally, we will use $w_i, d_i$, $b_i$  and $b_i^\prime$ to refer to the candidates based on the variable $x_i$ and $e_j$ to refer to the candidate based on the clause $c_j$. Among these candidates, the $w_i$'s and $d_i$'s are ``dummy'' candidates, while the $b_i$'s correspond to $x_i$ and $b_i^\pr$ corresponds to $\overline{x_i}$. To recap, the set of candidates are given by:

$$\CC = \{ w_i, d_i, b_i, b_i^\pr ~|~ x_i \in \VV\} \cup \{e_j ~|~ c_j \in \TT \} \cup \{w, g\}.$$

Consider an arbitrary but fixed ordering over $\CC$, such as the lexicographic order. In this proof, the notation $\overrightarrow{\CC^\pr}$ for any $\CC^\pr \subseteq \CC$ will be used to denote the lexicographic ordering restricted to the subset $\CC^\pr$. Let $m$ denote $|\CC| = 3n + t + 2$, and let $\overrightarrow{s_m}=(2, 1, \ldots, 1, 0)\in\mathbb{N}^m$. 

For every variable $x_i \in \VV$, we introduce the following complete and partial votes.

$$\aaa_i := w_i \succ \overrightarrow{\CC \setminus \{w_i,b_i,d_i\}} \succ d_i \succ b_i \mbox{ and } \aaa_i^\pr := \aaa_i \setminus \{(b_i,c) ~|~ \mbox{ for all } c \in \CC \setminus \{b_i\}\}$$

$$\bbb_i := w_i \succ \overrightarrow{\CC \setminus \{w_i,b_i^\pr,d_i\}} \succ d_i \succ b_i^\pr \mbox{ and } \bbb_i^\pr := \bbb_i \setminus \{(b_i^\pr,c) ~|~ \mbox{ for all } c \in \CC \setminus \{b_i^\pr\}\}$$

We use $\ell^\star$ to refer to the candidate $b_j$ if the literal is positive and $b_j^\pr$ if the literal is negated. For every clause $c_j \in \TT$ given by $c_j = \{\ell_1, \ell_2, \ell_3\}$, we introduce the following complete and partial votes.

$$\qqq_{j,1} := \overrightarrow{\CC \setminus \{e_j,\ell_1^\star\}} \succ e_j \succ \ell_1^\star \mbox{ and } \qqq_{j,1}^\pr := \qqq_{j,1} \setminus \{(e_j,\ell_1^\star)\}$$

$$\qqq_{j,2} := \overrightarrow{\CC \setminus \{e_j,\ell_2^\star\}} \succ e_j \succ \ell_2^\star \mbox{ and } \qqq_{j,2}^\pr := \qqq_{j,2} \setminus \{(e_j,\ell_2^\star)\}$$

$$\qqq_{j,3} := \overrightarrow{\CC \setminus \{e_j,\ell_3^\star\}}\succ e_j \succ \ell_3^\star \mbox{ and } \qqq_{j,3}^\pr := \qqq_{j,3} \setminus \{(e_j,\ell_3^\star)\}$$

Let us define the following sets of votes:

$$\PP = \left( \bigcup_{i=1}^n \aaa_i \right) \cup \left( \bigcup_{i=1}^n \bbb_i \right) \cup \left(\bigcup_{\substack{1 \leq j \leq t,}\\\substack{1\leq b \leq 3}} \qqq_{j,b} \right)$$

and 

$$\PP^\pr = \left( \bigcup_{i=1}^n \aaa_i^\pr \right) \cup \left( \bigcup_{i=1}^n \bbb_i^\pr \right)  \cup \left(\bigcup_{\substack{1 \leq j \leq t,}\\ \substack{1\leq b \leq 3}} \qqq_{j,b}^\pr \right) $$ 

There exists a set of complete votes \WW of size polynomial in $m$ with the following properties due to \Cref{score_gen}. Let $s^+: \CC \longrightarrow\NB$ be a function mapping candidates to their scores from the set of votes $\PP \cup\WW$. Then $\WW$ can be constructed to ensure that the following hold.

\begin{itemize}
	\item $s^+(w_i) = s^+(w) + 1$ for all $1 \leq i \leq n$. 
	\item $s^+(e_j) = s^+(w) + 1$ for all $1 \leq j \leq t$. 
	\item $s^+(b_i) = s^+(w) - 2$ for all $1 \leq i \leq n$. 
	\item $s^+(b_i^\pr) = s^+(w) - 2$ for all $1 \leq i \leq n$. 
	\item $s^+(g) < s^+(w)$ and $s^+(d_i) < s^+(w)$ for all $1 \leq i \leq n$.
\end{itemize}

We now define the instance $\II^\pr$ of \PW to be $(\CC, \PP^\pr \cup \WW, w)$. This completes the description of the reduction. Observe that all the partial votes either have at most $m-1$ undetermined pairs, as required. We now turn to a proof of the equivalence. Before we begin making our arguments, observe that since $w$ does not participate in any undetermined pairs of the votes in $\PP^\pr$, it follows that the score of $w$ continues to be $s^+(w)$ in any completion of $\PP^\pr$. The intuition for the construction, described informally, is as follows. The score of every ``clause candidate'' needs to decrease by at least one, which can be achieved by pushing it down against its literal partner in the $\qqq_j$-votes. Also, the score of every $w_i$ must also decrease by at least one, and the only way to achieve this is to push either $b_i$ or $b_i^\pr$ to the top in the two votes corresponding to the variable $x_i$. This causes the candidate $b_i$ (or $b_i^\pr$, as the case may be) to gain a score of two, leading to a tie with $w$, and rendering it impossible for us to use it to ``fix'' the situation for a clause candidate. Therefore, in any successful completion, whether $b_i$ or $b_i^\pr$ retains the zero-position works as a signal for how the corresponding variable should be set by a satisfying assignment.

We now turn to a formal proof. In the forward direction, let $\tau: \VV \rightarrow \{0,1\}$ be a satisfying assignment for $\II$. Then we have the following completions of the votes in $\PP^\pr$. To begin with, for all $1\leq i \leq n$, we have: 

\begin{equation*}
  \aaa_i^{\pr\pr} := \left\{
  \begin{array}{rl}
w_i \succ \overrightarrow{\CC \setminus \{w_i,b_i,d_i\}} \succ d_i \succ b_i & \text{if } \tau(x_i) = 1,\\ 
b_i \succ \overrightarrow{\CC \setminus \{w_i,b_i,d_i\}} \succ d_i \succ w_i & \text{if } \tau(x_i) = 0.
\end{array} \right.	
\end{equation*}

and also:

\begin{equation*}
  \bbb_i^{\pr\pr} := \left\{
  \begin{array}{rl}
w_i \succ \overrightarrow{\CC \setminus \{w_i,b_i^\pr,d_i\}} \succ d_i \succ b_i^\pr & \text{if } \tau(x_i) = 0,\\ 
b_i^\pr \succ \overrightarrow{\CC \setminus \{w_i,b_i^\pr,d_i\}} \succ d_i \succ w_i & \text{if } \tau(x_i) = 1.
\end{array} \right.	
\end{equation*}

For a clause $c_j = \{\ell_1, \ell_2, \ell_3\}$, suppose $\tau(\ell_1) = 1$. Then we have the following completions for the votes $\qqq_{j,b}$, $1 \leq b \leq 3$:

$$\qqq^{\pr\pr}_{j,1} := \overrightarrow{\CC \setminus \{e_j,\ell_1^\star\}} \succ \ell_1^\star \succ e_j, $$

$$\qqq^{\pr\pr}_{j,2} := \overrightarrow{\CC \setminus \{e_j,\ell_2^\star\}} \succ e_j \succ \ell_2^\star, $$

$$\qqq^{\pr\pr}_{j,3} := \overrightarrow{\CC \setminus \{e_j,\ell_3^\star\}} \succ e_j \succ \ell_3^\star$$

The completions for the cases when $\tau(\ell_2) = 1$ or $\tau(\ell_3) = 1$ are analogously defined. It is easily checked that $w$ is a co-winner in this completion, because the score of every $b_i$ and $b_i^\pr$ increases by at most two (given that we based the extensions on a satisfying assignment), and the scores of the $w_i$'s and the $e_j$'s decrease by one, as required. 

This completes the forward direction of the argument. In the other direction, let $\PP^{\pr\pr}$ be any completion of the votes in $\PP^\pr$ which makes $w$ a co-winner with respect to $s$. Let $s^\star$ be the function that computes the scores of all the candidates with respect to $\PP^{\pr\pr}$. We define the following assignment to the variables of $\II$ based on $\PP^{\pr\pr}$:

\begin{equation*}
  \tau(x_i) := \left\{
  \begin{array}{rl}
1 & \text{if } w_i \succ \overrightarrow{\CC \setminus \{w_i,b_i,d_i\}} \succ d_i \succ b_i \in \PP^{\pr\pr},
\\ 
0 & \text{if } w_i \succ \overrightarrow{\CC \setminus \{w_i,b_i,d_i\}} \succ d_i \succ b_i^\pr \in \PP^{\pr\pr}.
\end{array} \right.	
\end{equation*}

We claim that $\tau$, as defined above, satisfies every clause in $\II$. Consider any clause $c_j \in \TT$. Observe that the score of the corresponding candidate, $e_j$, must decrease by at least one in any valid completion, since $s^+(e_j) = s^+(w) + 1$. notice that any completion of the votes corresponding to the variables $x_i$ cannot influence the score of $e_j$, because the only candidates that change scores in any completion are $b_i, b_i^\pr, w_i$ and $d_i$. Therefore, in at least one of the votes $\qqq_{j,b}$, $1 \leq b \leq 3$, we must have a completion where $e_j$ appears at the last position. We claim that the literal $\ell$ that consequently appears at position $q$ must be set to one by $\tau$. Indeed, suppose not. Then we have two cases, as follows:

\begin{itemize}
	\item Suppose the literal $\ell$ corresponds to the positive appearance of a variable $x_j$. If $\tau(x_j) = 0$, then this implies that $b_i \succ \overrightarrow{\CC \setminus \{w_i,b_i,d_i\}} \succ d_i \succ w_i \in \PP^{\pr\pr}$ (if not, then the score of $w_i$ remains unchanged, a contradiction). However, this implies that $b_i$ has gained a score of three altogether, which is also a contradiction.
	\item Suppose the literal $\ell$ corresponds to the negated appearance of a variable $x_j$. If $\tau(x_j) = 1$, then this implies that $b_i^\pr \succ \overrightarrow{\CC \setminus \{w_i,b_i,d_i\}} \succ d_i \succ w_i \in \PP^{\pr\pr}$ (if not, then the score of $w_i$ remains unchanged, a contradiction). However, this implies that $b_i^\pr$ has gained a score of three altogether, which is also a contradiction.
\end{itemize}

 {\bf Proof of Part 3.} 
Now we turn to the proof of the polynomial time solvable case. Let the input instance of \PW be $(\CC, \PP, c)$ where every partial vote in \PP has at most $m-2$ pairs of candidates whose ordering is undetermined. For any $\ppp \in \PP$, let  $A(\ppp) \subseteq \CC$ denote the set of candidates $x$ for which $(y \succ x) \notin \ppp$ for any $y \in \CC$. Note that in any valid extension of $\ppp$, the candidate who occupies the first position (thereby getting a score of two) belongs to $A(\ppp)$. Similarly, let $B(\ppp) \subseteq \CC$ denote the set of candidates $x$ for which $(x \succ y) \notin \ppp$ for any $y \in \CC$. Note that in any valid extension of $\ppp$, the candidate who occupies the last position (thereby getting a score of zero) belongs to $B(\ppp)$. Also, since there are at most $m-2$ missing pairs, note that $A(\ppp) \cap B(\ppp) = \emptyset$. 

In every partial vote in \PP we place the candidate $c$ as high as possible. Suppose in a partial vote \ppp in \PP, one undetermined pair of candidates appears at positions $i$ and $i+1$ (from the bottom), where $i+1$ is not the top position and $i$ is not the bottom position. Then we fix the ordering of the undetermined pair of candidates in \ppp arbitrarily. Let us call the resulting profile $\PP^\pr$. It is easy to see that $(\CC, \PP, c)$ is a \YES instance if and only if $(\CC, \PP^\pr, c)$ is a \YES instance. Notice that the position of $c$ in every vote in $\PP^\pr$ is fixed and thus we know the score of $c$; let it be $s(c)$. Also we can compute the minimum score that every candidate receives over all extensions of $\PP^\pr$. Let $s(w)$ be the minimum score of candidate $w$. If there exists a candidate $z$ such that $s(z)>s(c)$, then we output \NO. 

Otherwise, we construct the following flow graph $\GG=(\VV, \EE)$. For every partial vote \ppp in $\PP^\pr$, we add two vertices $a_\ppp$ and  $b_\ppp$ in \VV. We also add a vertex $v_w$ in \VV for every candidate $w$ other than $c$. We also add two special vertices $s$ and $t$ in \VV. We add an edge from $s$ to $a_\ppp$ for every $\ppp\in\PP^\pr$ for which $A(\ppp)$ is non-empty, and the capacity of this edge is $1$. We also add an edge from $s$ to $b_\ppp$ for all $\ppp \in \PP^\pr$ for which $B(\ppp)$ is non-empty, and the capacity of these edges is equal to $|B(\ppp)| - 1$. For every vote $\ppp$, we add an edge from the vertex $a_\ppp$ to all vertices in $A(v_\ppp)$, and an edge from the vertex $b_\ppp$ to all vertices in $B(v_\ppp)$. All these edges have a capacity of one. Finally, we an edge from $v_w$ to $t$ of capacity $s(c)-s(w)$ for every candidate $w$ other than $c$. 

Let the number of votes in $\PP^\pr$ which are not complete be $t$. Now it is easy to see that the $(\CC, \PP^\pr, c)$ is a \YES instance if and only if there is a flow of size $t + \sum_{\ppp \in \VV^\pr} (|B(\ppp)| - 1)$ in \GG, where $\VV^\pr$ denotes the subset of votes who admit a non-empty $B$-set.
  
 {\bf Proof of Part 4.} Observe that if the difference vector has at least three $1$s, then the scoring rule is always either $\langle 1,1 \rangle$-contaminated or $\langle 0,1,0 \rangle$-contaminated. If the difference vector has at least two $1$s, then the scoring rule is either $\langle 0,1,0 \rangle$-contaminated or it is equivalent to $(2,1,\ldots,1,0)$. If the scoring rule has one $1$, then it is either plurality or veto or $k$-approval for some $1 < k < m-1$. Now the results follows from the fact that the $k$-approval voting rule is $\langle 0,1,0 \rangle$-contaminated for every $1 < k < m-1$.
\end{proof}

\CopelandTwo*

\begin{proof}
 The \PW problem for the Copeland$^\alpha$ voting rule is clearly in \NP. To prove \NP-hardness of \PW, we reduce \PW from \TDM. Let $\II = (\XX\cup\YY\cup\ZZ, \SS)$ be an arbitrary instance of \TDM. Let $|\XX|=|\YY|=|\ZZ|= m$. We construct an instance $\II^\pr$ of \PW from \II as follows.
 
 \[\text{Set of candidates: } \CC = \XX\cup\YY\cup\ZZ\cup\{c\}\cup\GG, \text{ where } \GG = \{g_1, \ldots, g_{10m}\} \]
 
 For every $\sss=(x, y, z)\in\SS$, let us consider the following vote $\ppp_\sss$.
 
 \[ \ppp_\sss = \overrightarrow{(\CC\setminus\{x, y, z\})_\sss} \suc x \suc y \suc z, \text{where } \overrightarrow{(\CC\setminus\{x, y, z\})_\sss} \text{ is any fixed ordering of } \CC\setminus\{x, y, z\} \]
 
 Using $\ppp_\sss$, we define a partial vote $\ppp_\sss^\pr$ as follows.
 
 \[ \ppp_\sss^\pr = \ppp_\sss \setminus \{(x,y), (x,z)\} \]
 
 Let us define $\PP = \cup_{\sss\in\SS} \ppp_\sss$ and $\PP^\pr = \cup_{\sss\in\SS} \ppp_\sss^\pr$. There exists a set of complete votes \QQ of size polynomial in $m$ with the following properties~\cite{mcgarvey1953theorem}.
 
 \begin{itemize}
  \item $\DD_{\PP\cup\QQ} (x,y) = \DD_{\PP\cup\QQ} (x, z) = 1, \forall x\in\XX, y\in\YY, z\in\ZZ$
  \item $\DD_{\PP\cup\QQ} (x, g_i) = 1, \DD_{\PP\cup\QQ} (g_j, x) = 1, \forall x\in\XX, i\in[8m+1], j\in[10m]\setminus[8m+1]$
  \item $\DD_{\PP\cup\QQ} (y, g_i) = \DD_{\PP\cup\QQ} (g_j, y) =\DD_{\PP\cup\QQ} (z, g_i) = \DD_{\PP\cup\QQ} (g_j, z)= 1, \forall y\in\YY, z\in\ZZ, i\in [10m-2], j\in\{10m-1, 10m\}$
  \item $\DD_{\PP\cup\QQ} (x, c) = \DD_{\PP\cup\QQ} (y, c) = \DD_{\PP\cup\QQ} (z, c) = \DD_{\PP\cup\QQ} (c, g) = 1, \forall x\in\XX, y\in\YY, z\in\ZZ, g\in\GG$
  \item $\DD_{\PP\cup\QQ} (g_j, g_i) = 1, \forall i\in[5m], j\in\{i+1, i+2, \ldots, i+\lfloor\nfrac{(10m-1)}{2}\rfloor\}$
 \end{itemize}
 
 All the pairwise margins which are not specified above is any integer in $\{-1, 1\}$. We summarize the Copeland score of every candidate in \CC from $\PP \cup \QQ$ in \Cref{tbl:cop_initial_appendix}. We now define the instance $\II^\pr$ of \PW to be $(\CC, \PP^\pr \cup \QQ, c)$. Notice that the number of undetermined pairs of candidates in every vote in $\II^\pr$ is at most $2$. This finishes the description of the \PW instance $\II^\pr$. Notice that since the number of voters in $\II^\pr$ is odd (since the pairwise margins are odd integers), the actual value of $\alpha$ does not play any role since no two candidates tie. Hence, in the rest of the proof, we omit $\alpha$ while mentioning the voting rule. We claim that \II and $\II^\pr$ are equivalent.
 
 \begin{table}[!htbp]
  \centering
  \begin{tabular}{|c|c|c|}\hline
   Candidates & Copeland score & Winning against\\\hline
   $c$ & $10m$ & \GG\\\hline
   $x\in\XX$ & $10m+2$ & $c$, \YY, \ZZ, $\{g_i: i\in [8m+1]\}$\\\hline
   $y\in\YY, z\in\ZZ$ & $10m-1$ & $c$, $\{g_i: i\in [10m-2]\}$\\\hline
   $g_i\in\GG$ & $<9m$ & $\subseteq\CC\setminus\{g_j : j\in\{i+1, i+2, \ldots, i+\lfloor\nfrac{10m-1)}{2}\rfloor\}\}$\\\hline
  \end{tabular}
  \caption{Summary of initial Copeland scores of the candidates}\label{tbl:cop_initial_appendix}
 \end{table}
 
 In the forward direction, suppose that \II be a \YES instance of \TDM. Then there exists a collection of $m$ sets $\SS^\pr\subset\SS$ in \SS such that $\cup_{\AA\in\SS^\pr} \AA = \XX\cup\YY\cup\ZZ$. We extend the partial vote $\ppp_\sss^\pr$ to complete vote $\bar{\ppp}_\sss$ as follows for every $\sss\in\SS$.
 
 $$
 \bar{\ppp}_\sss = 
 \begin{cases}
  \overrightarrow{(\CC\setminus\{x, y, z\})_\sss} \suc y \suc z \suc x & \sss\in\SS^\pr\\
  \overrightarrow{(\CC\setminus\{x, y, z\})_\sss} \suc x \suc y \suc z & \sss\notin\SS^\pr
 \end{cases}
 $$
 
 We consider the extension of $\PP^\pr$ to $\bar{\PP} = \cup_{\sss\in\SS} \bar{\ppp}_\sss$. We observe that $c$ is a co-winner in the profile $\bar{\PP}\cup\QQ$ since the Copeland score of $c$, every $x\in\XX, y\in\YY$, and $z\in\ZZ$ in $\bar{\PP}\cup\QQ$ is $10m$ and the Copeland score of every candidate in \GG in $\bar{\PP}\cup\QQ$ is strictly less than $9m$.
 
 In the reverse direction we suppose that the \PW instance $\II^\pr$ be a \YES instance. Then there exists an extension of the set of partial votes $\PP^\pr$ to a set of complete votes $\bar{\PP}$ such that $c$ is a co-winner in $\bar{\PP}\cup\QQ$. Let us call the extension of the partial vote $\ppp_\sss^\pr$ in $\bar{\PP}$ $\bar{\ppp}_\sss$. First we notice that the Copeland score of $c$ in $\bar{\PP}\cup\QQ$ is $10m$ since the relative ordering of $c$ with respect to every other candidate is already fixed in $\PP^\pr\cup\QQ$. Now we observe that, in $\PP\cup\QQ$, the Copeland score of every candidate in \XX is $2$ more than the Copeland score of $c$, whereas the Copeland score of every candidate in \YY and \ZZ is $1$ less than the Copeland score of $c$. Hence, the only way for $c$ to co-win the election is as follows: every candidate in \XX loses against exactly one candidate in \YY and exactly one candidate in \ZZ. This in turn is possible only if, for every $x\in\XX$, there exists a unique $\sss = (x, y, z)\in\SS$ such that $\bar{\ppp}_\sss = \overrightarrow{(\CC\setminus\{x, y, z\})_\sss} \suc y \suc z \suc x$; we call that unique \sss corresponding to every $x\in\XX$ $\sss_x$. We now claim that $\TT = \{\sss_x: x\in\XX\}$ forms a three dimensional matching of $\II^\pr$. First notice that, $|\TT|=m$ since there is exactly one $\sss_x$ for every $x\in\XX$. If \TT does not form a three dimensional matching of $\II^\pr$, then there exists a candidate in $\YY\cup\ZZ$ whose Copeland score is strictly more than the Copeland score of $c$ (which is $10m$). However, this contradicts our assumption that $c$ is a co-winner in $\bar{\PP}\cup\QQ$. Hence \TT forms a three dimensional matching of $\II$ and thus $\II$ is a \YES instance.
\end{proof}

\CopelandPoly*

\begin{proof}
 Let us prove the result for $\alpha=0$. The proof for $\alpha=1$ case is similar. Let the input instance of \PW be $(\CC, \PP, c)$ where every partial vote in \PP has at most one pair of candidates whose ordering is undetermined. We consider an extension $\PP^\pr$ of \PP where the candidate $c$ is placed as high as possible. For every two candidates $x, y\in\CC$, let $\VV_{\{x, y\}}$ be the set of partial votes in $\PP^\pr$ where the ordering of $x$ and $y$ is undetermined. Let $\BB$ be the set of pairs of vertices $\{x,y\}$ for which it is possible to make $x$ tie with $y$ by fixing the ordering of $x$ and $y$ in the votes in $\VV_{\{x,y\}}$. For every $\{x, y\}\in\BB$, we also fix the orderings of $x$ and $y$ in $\PP^\pr$ in such a way that $x$ and $y$ tie. We first observe that the \PW instance $(\CC, \PP, c)$ is a \YES instance if and only if $(\CC, \PP^\pr, c)$ is a \YES instance since every vote in \PP has at most one pair of candidates whose ordering is undetermined. Let the Copeland score of $c$ in $\PP^\pr$ be $s(c)$. We put every unordered pair of candidates $\{x, y\}\subset\CC\setminus\{c\}$ in a set \AA if setting $x$ preferred over $y$ in every vote $\VV_{\{x, y\}}$ makes $x$ defeat $y$ and setting $y$ preferred over $x$ in every vote in $\VV_{\{x,y\}}$ makes $y$ defeat $x$ in pairwise election. Note that \AA can be computed in polynomial amount of time. Now we construct the following instance $\II = (\GG = (\UU, \EE), s, t)$ of the maximum $s-t$ flow problem. The vertex set \UU of \GG consists of two special vertices $s$ and $t$, one vertex $u_{\{x,y\}}$ for every $\{x,y\}$ in \AA, one vertex $u_a$ for every candidate $a\in\CC$. For every candidate $x\in\CC\setminus\{c\}$, let $n_x$ be the number of candidates in \CC whom $x$ defeats pairwise in every extension of $\PP^\pr$. Observe that $n_x$ can be computed in polynomial amount of time. We answer \NO if there exists a $x\in\CC\setminus\{c\}$ whose $n_x > s(c)$ since the Copeland score of $x$ is more than the Copeland score of $c$ in every extension of $\PP^\pr$ and thus $c$ cannot co-win. For every $\{x, y\}\in\AA$, we add one edge from $u_{\{x,y\}}$ to $x$, one edge from $u_{\{x,y\}}$ to $y$, and one edge from $s$ to $u_{\{x,y\}}$ each with capacity $1$. For every vertex $u_x$ with $n_x < s(c)$, we add an edge from $x$ to $t$ with capacity $s(c)-n_x$. We claim that the \PW instance $(\CC, \PP^\pr, c)$ is a \YES instance if and only if there is a flow from $s$ to $t$ in \GG of size $\sum_{x\in\CC} (n_x-s(c))$. The proof of correctness follows easily from the construction of \GG.
\end{proof}

\CopelandZeroHalf*

\begin{proof}
 The \PW problem for the Copeland$^\alpha$ voting rule is clearly in \NP. To prove \NP-hardness of \PW, we reduce \PW from \SAT. Let $\II$ be an instance of \SAT, over the variables $\VV = \{x_1, \ldots, x_n\}$ and with clauses $\TT = \{c_1, \ldots, c_m\}$. We construct an instance $\II^\pr$ of \PW from \II as follows. 
 \[\text{Set of candidates: } \CC = \{x_i, \bar{x}_i, d_i: i\in[n]\}\cup\{c_i:i\in[m]\}\cup\{c\}\cup\GG, \text{ where } \GG = \{g_1, \ldots, g_{mn}\} \]
 
 For every $i\in[n]$, let us consider the following votes $\ppp_{x_i}^1, \ppp_{x_i}^2, \ppp_{\bar{x}_i}^1, \ppp_{\bar{x}_i}^2$. 
 \[ \ppp_{x_i}^1, \ppp_{x_i}^2: x_i\suc d_i\suc \text{others}, \ppp_{\bar{x}_i}^1, \ppp_{\bar{x}_i}^2: \bar{x}_i\suc d_i\suc \text{others} \]
 
 Using $\ppp_{x_i}^1, \ppp_{x_i}^2, \ppp_{\bar{x}_i}^1, \ppp_{\bar{x}_i}^2$, we define the partial votes $\ppp_{x_i}^{1\pr}, \ppp_{x_i}^{2\pr}, \ppp_{\bar{x}_i}^{1\pr}, \ppp_{\bar{x}_i}^{2\pr}$ as follows. 
 \[ \ppp_{x_i}^{1\pr}, \ppp_{x_i}^{2\pr}: \ppp_{x_i}^1 \setminus \{(x_i, d_i)\}, \ppp_{\bar{x}_i}^{1\pr}, \ppp_{\bar{x}_i}^{2\pr}: \ppp_{\bar{x}_i}^1 \setminus \{(\bar{x}_i, d_i)\}\]
 
 Let a clause $c_j$ involves the literals $\el_j^1, \el_j^2, \el_j^3$. For every $j\in[m]$, let us consider the following votes $\qqq_j({\el_j^1}), \qqq_j({\el_j^2}), \qqq_j({\el_j^3})$. 
 \[ \qqq_j({\el_j^k}): c_j\suc \el_j^k\suc \text{others}, \forall k\in[3] \]
 
 Using $\qqq_j({\el_j^1}), \qqq_j({\el_j^2}), \qqq_j({\el_j^3})$, we define the partial votes $\qqq_j^\pr({\el_j^1}), \qqq_j^\pr({\el_j^2}), \qqq_j^\pr({\el_j^3})$ as follows. 
 \[ \qqq_j^\pr({\el_j^k}): \qqq_j({\el_j^k})\setminus\{(c_j, \el_j^k)\}, \forall k\in[3] \]
 
 Let us define $\PP = \cup_{i\in[n]} \{\ppp_{x_i}^1, \ppp_{x_i}^2, \ppp_{\bar{x}_i}^1, \ppp_{\bar{x}_i}^2\} \cup_{j\in[m]} \{\qqq_j({\el_j^1}), \qqq_j({\el_j^2}), \qqq_j({\el_j^3})\}$ and $\PP^\pr = \cup_{i\in[n]} \{\ppp_{x_i}^{1\pr}, \ppp_{x_i}^{2\pr}, \ppp_{\bar{x}_i}^{1\pr}, \ppp_{\bar{x}_i}^{2\pr}\} \cup_{j\in[m]} \{\qqq_j^\pr({\el_j^1}), \qqq_j^\pr({\el_j^2}), \qqq_j^\pr({\el_j^3})\}$. \shortversion{There exists a set of complete votes \QQ of size polynomial in $n$ and $m$ which realizes \Cref{tbl:cop_initial_alpha_zero_half_appendix}~\cite{mcgarvey1953theorem}. All the wins and defeats in \Cref{tbl:cop_initial_alpha_zero_half_appendix} are by a margin of $2$.} We now define the instance $\II^\pr$ of \PW to be $(\CC, \PP^\pr \cup \QQ, c)$. Notice that the number of undetermined pairs of candidates in every vote in $\II^\pr$ is at most $1$. This finishes the description of the \PW instance. We claim that \II and $\II^\pr$ are equivalent. 
 
 \begin{table}[!htbp]
  \centering
  \resizebox{\textwidth}{!}{  
  \begin{tabular}{|c|c|c|c|c|}\hline
   Candidates & Copeland$^\alpha$ score & Winning against & Losing against & Tie with\\\hline
   
   $c$ & \makecell{$(2n+m)\alpha$\\$ + n + \nfrac{3mn}{4}$} & $G^\pr\subset\GG, |G^\pr| = n + \nfrac{3mn}{4}$ & \makecell{$\GG\setminus G^\pr,|G^\pr|=n + \nfrac{3mn}{4}$ \\$ d_i, \forall i\in[n]$} & \makecell{$x_i, \bar{x}_i\forall i\in[n]$ \\ $c_j \forall j\in[m]$}\\\hline
   
   $x_i, \forall i\in[n]$ & \makecell{$(2n+m)\alpha$\\$ + n + \nfrac{3mn}{4}$} & \makecell{$G^\prr\subset\GG, |G^\prr| = \nfrac{3mn}{4}$ \\$d_i \forall i\in[n]$} & $\GG\setminus (G^\pr\cup G^\prr)$ & \makecell{$c, G^\pr\subset\GG, |G^\pr| = m $ \\$x_j,   
   \forall j\in[n]\setminus\{i\}$\\$\bar{x}_j \forall j\in[n]$} \\\hline
   
   $\bar{x}_i, \forall i\in[n]$ & \makecell{$(2n+m)\alpha$\\$ + n + \nfrac{3mn}{4}$} & \makecell{$G^\prr\subset\GG, |G^\prr| = \nfrac{3mn}{4}$ \\$d_i \forall i\in[n]$} & $\GG\setminus (G^\pr\cup G^\prr)$ & \makecell{$c, G^\pr\subset\GG, |G^\pr| = m $ \\$\bar{x}_j, \forall j\in[n]\setminus\{i\}$\\$x_j \forall j\in[n]$} \\\hline
   
   $c_j, \forall j\in[m]$ & \makecell{$(2n+m-1)\alpha $\\$+ n + \nfrac{3mn}{4}+1$} & \makecell{$x_i, \bar{x}_i\forall i\in[n]$\\$G^\pr\subset\GG, |G^\pr| = \nfrac{3mn}{4}-n+1$} & \makecell{$\GG\setminus (G^\pr\cup G^\prr)$ \\$ d_i, \forall i\in[n]$} & \makecell{$c$\\$c_j \forall j\in[m]\setminus\{i\}$\\$G^\prr\subset\GG, |G^\prr| = 2n-1$} \\\hline
   
   $d_i, i\in[n]$& \makecell{$(2n+m)\alpha$\\$ + n + \nfrac{3mn}{4}-1$} &\makecell{$c, c_j, \forall j\in[m]$\\$G^\prr\subset\GG, |G^\prr| = \nfrac{3mn}{4} - m + n -2$} & \makecell{$x_i, \bar{x}_i\forall i\in[n]$\\$\GG\setminus (G^\pr\cup G^\prr)$} & $G^\pr\subset\GG, |G^\pr| = 2n+m$ \\\hline
   
   $g_i, \forall i\in[mn]$ & $ < \nfrac{3mn}{4}$ &  & $\forall j\in\{i+k: k\in [\lfloor\nfrac{(mn-1)}{2}\rfloor]$ &  \\\hline
  \end{tabular}}
  \caption{Summary of initial Copeland$^\alpha$ scores of the candidates. All the wins and defeats in the table are by a margin of $2$.}\label{tbl:cop_initial_alpha_zero_half_appendix}
 \end{table}
 
 In the forward direction, suppose that \II be a \YES instance of \SAT. Then there exists an assignment $x_i^*$ of variables $x_i$ for all $i\in[n]$ to $0$ or $1$ that satisfies all the clauses $c_j, j\in[m]$. For every $i\in[n]$, we extend the partial votes $\ppp_{x_i}^{1\pr}, \ppp_{x_i}^{2\pr}, \ppp_{\bar{x}_i}^{1\pr}, \ppp_{\bar{x}_i}^{2\pr}$ to the complete votes $\bar{\ppp}_{x_i}^{1}, \bar{\ppp}_{x_i}^{2}, \bar{\ppp}_{\bar{x}_i}^{1}, \bar{\ppp}_{\bar{x}_i}^{2}$ as follows.
 $$
 \bar{\ppp}_{x_i}^{1}, \bar{\ppp}_{x_i}^{2} = 
 \begin{cases}
  x_i\suc d_i\suc \text{others} & x_i^*=0\\
  d_i\suc x_i\suc \text{others} & x_i^*=1
 \end{cases}~;~  \bar{\ppp}_{\bar{x}_i}^{1}, \bar{\ppp}_{\bar{x}_i}^{2} = 
 \begin{cases}
  \bar{x}_i\suc d_i\suc \text{others} & x_i^*=1\\
  d_i\suc \bar{x}_i\suc \text{others} & x_i^*=0
 \end{cases}
 $$
 
 Let $c_j$ be a clause involving literals $\el_j^1, \el_j^2, \el_j^3$ and let us assume, without loss of generality, that the assignment $\{x_i^*\}_{i\in[n]}$ makes the literal $\el_j^3$ $1$. For every $j\in[m]$, we extend the partial votes $\qqq_j^\pr({\el_j^1}), \qqq_j^\pr({\el_j^2}), \qqq_j^\pr({\el_j^3})$ to the complete votes $\bar{\qqq}_j({\el_j^1}), \bar{\qqq}_j({\el_j^2}), \bar{\qqq}_j({\el_j^3})$ as follows. 
 \[ \bar{\qqq}_j({\el_j^3}) = \el_j^3 \suc c_j\suc \text{others}, \bar{\qqq}_j({\el_j^k}) = c_j\suc \el_j^k\suc \text{others}, \forall k\in[2] \]
 
 We consider the extension of $\PP^\pr$ to $\bar{\PP} = \cup_{i\in[n]} \{\bar{\ppp}_{x_i}^1, \bar{\ppp}_{x_i}^2, \bar{\ppp}_{\bar{x}_i}^1, \bar{\ppp}_{\bar{x}_i}^2\} \cup_{j\in[m]} \{\bar{\qqq}_j({\el_j^1}), \bar{\qqq}_j({\el_j^2}), \bar{\qqq}_j({\el_j^3})\}$. We observe that $c$ is a co-winner in the profile $\bar{\PP}\cup\QQ$ since the Copeland$^\alpha$ score of $c$, $d_i$ for every $i\in[n]$, and $c_j$ for every $j\in[m]$ in $\bar{\PP}\cup\QQ$ is $(2n+m)\alpha + n + \nfrac{3mn}{4}$, the Copeland$^\alpha$ score of $x_i$ and $\bar{x_i}$ for every $i\in[n]$ is at most $(2n+m)\alpha + n + \nfrac{3mn}{4}$ since every literal appears in at most two clauses and $\alpha\le \nfrac{1}{2}$, and the Copeland$^\alpha$ score of the candidates in \GG in $\bar{\PP}\cup\QQ$ is strictly less than $\nfrac{3mn}{4}$.
 
 In the reverse direction we suppose that the \PW instance $\II^\pr$ be a \YES instance. Then there exists an extension of the set of partial votes $\PP^\pr$ to a set of complete votes $\bar{\PP}$ such that $c$ is a co-winner in $\bar{\PP}\cup\QQ$. Let us call the extension of the partial votes $\ppp_{x_i}^{1\pr}, \ppp_{x_i}^{2\pr}, \ppp_{\bar{x}_i}^{1\pr}, \ppp_{\bar{x}_i}^{2\pr}$ in $\bar{\PP}$ $\bar{\ppp}_{x_i}^{1}, \bar{\ppp}_{x_i}^{2}, \bar{\ppp}_{\bar{x}_i}^{1}, \bar{\ppp}_{\bar{x}_i}^{2}$ and the extension of the partial votes $\qqq_j^\pr({\el_j^1}), \qqq_j^\pr({\el_j^2}), \qqq_j^\pr({\el_j^3})$ in $\bar{\PP}$ $\bar{\qqq}_j({\el_j^1}), \bar{\qqq}_j({\el_j^2}), \bar{\qqq}_j({\el_j^3})$. Now we notice that the Copeland$^\alpha$ score of $c$ in $\bar{\PP}\cup\QQ$ is $(2n+m)\alpha + n + \nfrac{3mn}{4}$ since the relative ordering of $c$ with respect to every other candidate is already fixed in $\PP^\pr\cup\QQ$. We observe that the Copeland$^\alpha$ score of $d_i$ for every $i\in[n]$ can increase by at most $1$ from $\PP \cup \QQ$ without defeating $c$. Hence it cannot be the case that $d_i$ is preferred over $x_i$ in both $\bar{\ppp}_{x_i}^{1}$ and $\bar{\ppp}_{x_i}^{2}$ and $d_i$ is preferred over $\bar{x}_i$ in both $\bar{\ppp}_{\bar{x}_i}^{1}$ and $\bar{\ppp}_{\bar{x}_i}^{2}$. We define $x_i^*$ to be $1$ if $d_i$ is preferred over $x_i$ in both $\bar{\ppp}_{x_i}^{1}$ and $\bar{\ppp}_{x_i}^{2}$ and $0$ otherwise. We claim that $\{x_i^*\}_{i\in[n]}$ is a satisfying assignment to all the clauses in \TT. Suppose not, then there exists a clause $c_i$ which is not satisfied by the assignment$\{x_i^*\}_{i\in[n]}$. Hence, for $c$ to co-win in $\bar{\PP}\cup\QQ$, the Copeland$^\alpha$ score of $c_j$ for every $j\in[m]$ must decrease by at least $(1-\alpha)$ from $\PP \cup \QQ$. Now let us consider the candidate $c_i$. Hence there must be a candidate $\el_i$ such that the literal $\el_i$ appear in the clause $c_i$ and the candidate $\el_i$ is preferred over the candidate $c_i$ in $\bar{\qqq}_i({\el_i})$. However, this increases the score of $\el_i$ by $\alpha$. Also, since the assignment $\{x_i^*\}_{i\in[n]}$ makes $\el_i$ false (by our assumption, the clause $c_i$ is not satisfied), the Copeland$^\alpha$ score of $\el_i$ in $\bar{\PP}\cup\QQ$ is strictly more than $(2n+m)\alpha + n + \nfrac{3mn}{4}$ since $\alpha>0$. This contradicts our assumption that $c$ co-wins in $\bar{\PP}\cup\QQ$. Hence $\{x_i^*\}_{i\in[n]}$ is a satisfying assignment of the clause in \TT and thus \II is a \YES instance.
\end{proof}

\CopelandHalfOne*

\begin{proof}
 The \PW problem for the Copeland$^\alpha$ voting rule is clearly in \NP. To prove \NP-hardness of \PW, we reduce \PW from \SAT. Let $\II$ be an instance of \SAT, over the variables $\VV = \{x_1, \ldots, x_n\}$ and with clauses $\TT = \{c_1, \ldots, c_m\}$. We construct an instance $\II^\pr$ of \PW from \II as follows. 
 \[\text{Set of candidates: } \CC = \{x_i, \bar{x}_i, d_i: i\in[n]\}\cup\{c_i:i\in[m]\}\cup\{c\}\cup\GG, \text{ where } \GG = \{g_1, \ldots, g_{mn}\} \]
 
 For every $i\in[n]$, let us consider the following votes $\ppp_{x_i}^1, \ppp_{x_i}^2, \ppp_{\bar{x}_i}^1, \ppp_{\bar{x}_i}^2$.
 
 \[ \ppp_{x_i}^1, \ppp_{x_i}^2: x_i\suc d_i\suc \text{others} \]
 \[ \ppp_{\bar{x}_i}^1, \ppp_{\bar{x}_i}^2: \bar{x}_i\suc d_i\suc \text{others} \]
 
 Using $\ppp_{x_i}^1, \ppp_{x_i}^2, \ppp_{\bar{x}_i}^1, \ppp_{\bar{x}_i}^2$, we define the partial votes $\ppp_{x_i}^{1\pr}, \ppp_{x_i}^{2\pr}, \ppp_{\bar{x}_i}^{1\pr}, \ppp_{\bar{x}_i}^{2\pr}$ as follows.
 
 \[ \ppp_{x_i}^{1\pr}, \ppp_{x_i}^{2\pr}: \ppp_{x_i}^1 \setminus \{(x_i, d_i)\}\]
 \[ \ppp_{\bar{x}_i}^{1\pr}, \ppp_{\bar{x}_i}^{2\pr}: \ppp_{\bar{x}_i}^1 \setminus \{(\bar{x}_i, d_i)\}\]
 
 Let a clause $c_j$ involves the literals $\el_j^1, \el_j^2, \el_j^3$. For every $j\in[m]$, let us consider the following votes $\qqq_j({\el_j^1}), \qqq_j({\el_j^2}), \qqq_j({\el_j^3})$.
 
 \[ \qqq_j({\el_j^k}): c_j\suc \el_j^k\suc \text{others}, \forall k\in[3] \]
 
 Using $\qqq_j({\el_j^1}), \qqq_j({\el_j^2}), \qqq_j({\el_j^3})$, we define the partial votes $\qqq_j^\pr({\el_j^1}), \qqq_j^\pr({\el_j^2}), \qqq_j^\pr({\el_j^3})$ as follows.
 
 \[ \qqq_j^\pr({\el_j^k}): \qqq_j({\el_j^k})\setminus\{(c_j, \el_j^k)\}, \forall k\in[3] \]
 
 Let us define $\PP = \cup_{i\in[n]} \{\ppp_{x_i}^1, \ppp_{x_i}^2, \ppp_{\bar{x}_i}^1, \ppp_{\bar{x}_i}^2\} \cup_{j\in[m]} \{\qqq_j({\el_j^1}), \qqq_j({\el_j^2}), \qqq_j({\el_j^3})\}$ and $\PP^\pr = \cup_{i\in[n]} \{\ppp_{x_i}^{1\pr}, \ppp_{x_i}^{2\pr}, \ppp_{\bar{x}_i}^{1\pr}, \ppp_{\bar{x}_i}^{2\pr}\} \cup_{j\in[m]} \{\qqq_j^\pr({\el_j^1}), \qqq_j^\pr({\el_j^2}), \qqq_j^\pr({\el_j^3})\}$. There exists a set of complete votes \QQ of size polynomial in $n$ and $m$ with the following properties~\cite{mcgarvey1953theorem}.

 \begin{itemize}
  \item Let $G_{\nfrac{3mn}{4}}\subset\GG$ such that $G_{\nfrac{3mn}{4}}=\nfrac{3mn}{4}$. Then we have $\forall i\in[n], \DD_{\PP\cup\QQ} (x_i, x_j) = \DD_{\PP\cup\QQ} (x_i, \bar{x}_k) = \DD_{\PP\cup\QQ} (x_i, c) = \DD_{\PP\cup\QQ} (x_i, c_{j^\pr}) = 0, \forall j\in[n]\setminus\{i\} \forall k\in[n], \forall j^\pr\in[m], \DD_{\PP\cup\QQ} (x_i, g) = 2, \DD_{\PP\cup\QQ} (x_i, d_k) = 2,  \DD_{\PP\cup\QQ} (x_i, g^\pr) = -2, \forall k\in[n] \forall g\in\GG_{\nfrac{3mn}{4}} \forall g^\pr\in\GG\setminus G_{\nfrac{3mn}{4}}$
  
  \item Let $G_{\nfrac{3mn}{4}}\subset\GG$ such that $G_{\nfrac{3mn}{4}}=\nfrac{3mn}{4}$. Then we have $\forall i\in[n], \DD_{\PP\cup\QQ} (\bar{x}_i, \bar{x}_j) = \DD_{\PP\cup\QQ} (\bar{x}_i, x_k) = \DD_{\PP\cup\QQ} (\bar{x}_i, c) = \DD_{\PP\cup\QQ} (\bar{x}_i, c_j) = 0, \forall j\in[n]\setminus\{i\} \forall k\in[n], \forall j\in[m], \DD_{\PP\cup\QQ} (\bar{x}_i, g) = 2, \DD_{\PP\cup\QQ} (\bar{x}_i, d_k) = 2,  \DD_{\PP\cup\QQ} (\bar{x}_i, g^\pr) = -2, \forall k\in[n] \forall g\in\GG_{\nfrac{3mn}{4}} \forall g^\pr\in\GG\setminus G_{\nfrac{3mn}{4}}$
  
  \item Let $G_{n+\nfrac{3mn}{4}}\subset\GG$ such that $|G_{n+\nfrac{3mn}{4}}|=n+\nfrac{3mn}{4}$. Then we have $\DD_{\PP\cup\QQ} (c, x_i) = \DD_{\PP\cup\QQ} (c, \bar{x}_i) = \DD_{\PP\cup\QQ} (c,c_j) = 0, \forall i\in[n] \forall j\in[m], \DD_{\PP\cup\QQ} (c, g) = \DD_{\PP\cup\QQ} (g^\pr, c) = 2, \forall g\in G_{n + \nfrac{3mn}{4}}, g^\pr\in\GG\setminus G_{n + \nfrac{3mn}{4}}$
  
  \item Let $G_{n+\nfrac{3mn}{4}}\subset\GG$ such that $|G_{n+\nfrac{3mn}{4}}|=n+\nfrac{3mn}{4}$, $\bar{g}\in\GG\setminus G_{n+\nfrac{3mn}{4}}$. Then we have $\forall i\in[m], \DD_{\PP\cup\QQ} (c_i, c_j) = \DD_{\PP\cup\QQ} (c_i, c) = \DD_{\PP\cup\QQ} (c_i, x_k) = \DD_{\PP\cup\QQ} (c_i, \bar{x}_k) = \DD_{\PP\cup\QQ} (c_i, \bar{g}) = 0, \forall j\in[m]\setminus\{i\} \forall k\in[n], \DD_{\PP\cup\QQ} (c_i, g) = \DD_{\PP\cup\QQ} (g^\prr, c_i) = \DD_{\PP\cup\QQ} (d_j, c_i) = 2, \forall g^\pr\in G_{n+\nfrac{3mn}{4}} \forall g^\prr\in\GG\setminus G_{n+\nfrac{3mn}{4}} \forall j\in[n]$
  
  \item Let $G_{2n+m}, G_{\nfrac{3mn}{4}-m+n-2}\subset\GG$ such that $|G_{2n+m}|=2n+m, G_{\nfrac{3mn}{4}-m+n-2}=\nfrac{3mn}{4}-m+n-2, G_{2n+m}\cap G_{\nfrac{3mn}{4}-m+n-2}=\emptyset$. Then we have $\forall i\in[n], \DD_{\PP\cup\QQ} (d_i, g) = 0, \forall g\in G_{2n+m}, \DD_{\PP\cup\QQ} (d_i, g^\pr) = \DD_{\PP\cup\QQ} (g^\prr, d_i) = 2, \forall g^\pr\in G_{\nfrac{3mn}{4}-m+n-2}, g^\prr\in \GG\setminus(G_{2n+m} \cup G_{\nfrac{3mn}{4}-m+n-2}) $
  
  \item $\forall i\in[mn], \DD_{\PP\cup\QQ} (g_j, g_i) = 2 \forall j\in\{i+k: k\in [\lfloor\nfrac{(mn-1)}{2}\rfloor]\}$
 \end{itemize}
 
 All the pairwise margins which are not specified above is $0$. We summarize the Copeland$^\alpha$ score of every candidate in \CC from $\PP \cup \QQ$ in \Cref{tbl:cop_initial_alpha_half_one_appendix}. We now define the instance $\II^\pr$ of \PW to be $(\CC, \PP^\pr \cup \QQ, c)$. Notice that the number of undetermined pairs of candidates in every vote in $\II^\pr$ is at most $1$. This finishes the description of the \PW instance. We claim that \II and $\II^\pr$ are equivalent.
 
 \begin{table}[!htbp]
  \centering
  \resizebox{\textwidth}{!}{  
  \begin{tabular}{|c|c|c|c|c|}\hline
   Candidates & Copeland$^\alpha$ score & Winning against & Losing against & Tie with\\\hline
   
   $c$ & \makecell{$(2n+m)\alpha$\\$ + n + \nfrac{3mn}{4}$} & $G^\pr\subset\GG, |G^\pr| = n + \nfrac{3mn}{4}$ & \makecell{$\GG\setminus G^\pr,|G^\pr|=n + \nfrac{3mn}{4}$ \\$ d_i, \forall i\in[n]$} & \makecell{$x_i, \bar{x}_i\forall i\in[n]$ \\ $c_j \forall j\in[m]$}\\\hline
   
   $x_i, \forall i\in[n]$ & \makecell{$(2n+m)\alpha$\\$ + n + \nfrac{3mn}{4}$} & \makecell{$G^\prr\subset\GG, |G^\prr| = \nfrac{3mn}{4}$ \\$d_i \forall i\in[n]$} & $\GG\setminus (G^\pr\cup G^\prr)$ & \makecell{$c, c_j \forall j\in[m]$ \\$x_j,   
   \forall j\in[n]\setminus\{i\}$\\$\bar{x}_j \forall j\in[n]$} \\\hline
   
   $\bar{x}_i, \forall i\in[n]$ & \makecell{$(2n+m)\alpha$\\$ + n + \nfrac{3mn}{4}$} & \makecell{$G^\prr\subset\GG, |G^\prr| = \nfrac{3mn}{4}$ \\$d_i \forall i\in[n]$} & $\GG\setminus (G^\pr\cup G^\prr)$ & \makecell{$c, c_j \forall j\in[m]$ \\$\bar{x}_j, \forall j\in[n]\setminus\{i\}$\\$x_j \forall j\in[n]$} \\\hline
   
   $c_j, \forall j\in[m]$ & \makecell{$(2n+m+1)\alpha $\\$+ n + \nfrac{3mn}{4}$} & \makecell{$G^\pr\subset\GG, |G^\pr| = \nfrac{3mn}{4}+n$} & \makecell{$\GG\setminus (G^\pr\cup G^\prr)$ \\$ d_i, \forall i\in[n]$} & \makecell{$c, x_i, \bar{x}_i\forall i\in[n]$\\$c_j \forall j\in[m]\setminus\{i\}$\\$G^\prr\subset\GG, |G^\prr| = 1$} \\\hline
   
   $d_i, i\in[n]$& \makecell{$(2n+m)\alpha$\\$ + n + \nfrac{3mn}{4}-1$} &\makecell{$c, c_j, \forall j\in[m]$\\$G^\prr\subset\GG, |G^\prr| = \nfrac{3mn}{4} - m + n -2$} & \makecell{$x_i, \bar{x}_i\forall i\in[n]$\\$\GG\setminus (G^\pr\cup G^\prr)$} & $G^\pr\subset\GG, |G^\pr| = 2n+m$ \\\hline
   
   $g_i, \forall i\in[mn]$ & $ < \nfrac{3mn}{4}$ &  & $\forall j\in\{i+k: k\in [\lfloor\nfrac{(mn-1)}{2}\rfloor]$ &  \\\hline
  \end{tabular}}
  \caption{Summary of initial Copeland$^\alpha$ scores of the candidates}\label{tbl:cop_initial_alpha_half_one_appendix}
 \end{table}
 
 In the forward direction, suppose that \II be a \YES instance of \SAT. Then there exists an assignment $x_i^*$ of variables $x_i$ for all $i\in[n]$ to $0$ or $1$ that satisfies all the clauses $c_j, j\in[m]$. For every $i\in[n]$, we extend the partial votes $\ppp_{x_i}^{1\pr}, \ppp_{x_i}^{2\pr}, \ppp_{\bar{x}_i}^{1\pr}, \ppp_{\bar{x}_i}^{2\pr}$ to the complete votes $\bar{\ppp}_{x_i}^{1}, \bar{\ppp}_{x_i}^{2}, \bar{\ppp}_{\bar{x}_i}^{1}, \bar{\ppp}_{\bar{x}_i}^{2}$ as follows.
 
 $$
 \bar{\ppp}_{x_i}^{1}, \bar{\ppp}_{x_i}^{2} = 
 \begin{cases}
  x_i\suc d_i\suc \text{others} & x_i^*=0\\
  d_i\suc x_i\suc \text{others} & x_i^*=1
 \end{cases}
 $$
 
 $$
 \bar{\ppp}_{\bar{x}_i}^{1}, \bar{\ppp}_{\bar{x}_i}^{2} = 
 \begin{cases}
  \bar{x}_i\suc d_i\suc \text{others} & x_i^*=1\\
  d_i\suc \bar{x}_i\suc \text{others} & x_i^*=0
 \end{cases}
 $$
 
 Let $c_j$ be a clause involving literals $\el_j^1, \el_j^2, \el_j^3$ and let us assume, without loss of generality, that the assignment ${x_i^*}_{i\in[n]}$ makes the literal $\el_j^3$ $1$. For every $j\in[m]$, we extend the partial votes $\qqq_j^\pr({\el_j^1}), \qqq_j^\pr({\el_j^2}), \qqq_j^\pr({\el_j^3})$ to the complete votes $\bar{\qqq}_j({\el_j^1}), \bar{\qqq}_j({\el_j^2}), \bar{\qqq}_j({\el_j^3})$ as follows.
 
 \[ \bar{\qqq}_j({\el_j^3}) = \el_j^3 \suc c_j\suc \text{others}, \bar{\qqq}_j({\el_j^k}) = c_j\suc \el_j^k\suc \text{others}, \forall k\in[2] \]
 
 We consider the extension of $\PP^\pr$ to $\bar{\PP} = \cup_{i\in[n]} \{\bar{\ppp}_{x_i}^1, \bar{\ppp}_{x_i}^2, \bar{\ppp}_{\bar{x}_i}^1, \bar{\ppp}_{\bar{x}_i}^2\} \cup_{j\in[m]} \{\bar{\qqq}_j({\el_j^1}), \bar{\qqq}_j({\el_j^2}), \bar{\qqq}_j({\el_j^3})\}$. We observe that $c$ is a co-winner in the profile $\bar{\PP}\cup\QQ$ since the Copeland$^\alpha$ score of $c$, $d_i$ for every $i\in[n]$, and $c_j$ for every $j\in[m]$ in $\bar{\PP}\cup\QQ$ is $(2n+m)\alpha + n + \nfrac{3mn}{4}$, the Copeland$^\alpha$ score of $x_i, \bar{x_i}$ for every $i\in[n]$ is at most $(2n+m)\alpha + n + \nfrac{3mn}{4}$ since every literal appears in at most two clauses and $1-\alpha\le \nfrac{1}{2}$, and the Copeland$^\alpha$ score of the candidates in \GG in $\bar{\PP}\cup\QQ$ is strictly less than $\nfrac{3mn}{4}$.
 
 In the reverse direction suppose the \PW instance $\II^\pr$ be a \YES instance. Then there exists an extension of the set of partial votes $\PP^\pr$ to a set of complete votes $\bar{\PP}$ such that $c$ is a co-winner in $\bar{\PP}\cup\QQ$. Let us call the extension of the partial votes $\ppp_{x_i}^{1\pr}, \ppp_{x_i}^{2\pr}, \ppp_{\bar{x}_i}^{1\pr}, \ppp_{\bar{x}_i}^{2\pr}$ in $\bar{\PP}$ $\bar{\ppp}_{x_i}^{1}, \bar{\ppp}_{x_i}^{2}, \bar{\ppp}_{\bar{x}_i}^{1}, \bar{\ppp}_{\bar{x}_i}^{2}$ and the extension of the partial votes $\qqq_j^\pr({\el_j^1}), \qqq_j^\pr({\el_j^2}), \qqq_j^\pr({\el_j^3})$ in $\bar{\PP}$ $\bar{\qqq}_j({\el_j^1}), \bar{\qqq}_j({\el_j^2}), \bar{\qqq}_j({\el_j^3})$. Now we notice that the Copeland$^\alpha$ score of $c$ in $\bar{\PP}\cup\QQ$ is $(2n+m)\alpha + n + \nfrac{3mn}{4}$ since the relative ordering of $c$ with respect to every other candidate is already fixed in $\PP^\pr\cup\QQ$. We observe that the Copeland$^\alpha$ score of $d_i$ for every $i\in[n]$ can increase by at most $1$ from $\PP \cup \QQ$. Hence it cannot be the case that $d_i$ is preferred over $x_i$ in both $\bar{\ppp}_{x_i}^{1}$ and $\bar{\ppp}_{x_i}^{2}$ and $d_i$ is preferred over $\bar{x}_i$ in both $\bar{\ppp}_{\bar{x}_i}^{1}$ and $\bar{\ppp}_{\bar{x}_i}^{2}$. We define $x_i^*$ to be $1$ if $d_i$ is preferred over $x_i$ in both $\bar{\ppp}_{x_i}^{1}$ and $\bar{\ppp}_{x_i}^{2}$ and $0$ otherwise. We claim that $\{x_i^*\}_{i\in[n]}$ is a satisfying assignment to all the clauses in \TT. Suppose not, then there exists a clause $c_i$ which is not satisfied by the assignment$\{x_i^*\}_{i\in[n]}$. The Copeland$^\alpha$ score of $c_j$ for every $j\in[m]$ in $\bar{\PP}\cup\QQ$ is $(2n+m+1)\alpha + n + \nfrac{3mn}{4}$. Hence, for $c$ to co-win in $\bar{\PP}\cup\QQ$, the Copeland$^\alpha$ score of $c_j$ for every $j\in[m]$ must decrease by at least $\alpha$ from $\PP \cup \QQ$. Now let us consider the candidate $c_i$. There must be a candidate $\el_i$ such that the literal $\el_i$ appear in the clause $c_i$ and $\el_i$ is preferred over the candidate $c_i$ in $\bar{\qqq}_i({\el_i})$. However, this increases the score of $\el_i$ by $\alpha$. Also, since the assignment $\{x_i^*\}_{i\in[n]}$ makes $\el_i$ false (since by assumption, the clause $c_i$ is not satisfied), the Copeland$^\alpha$ score of $\el_i$ in $\bar{\PP}\cup\QQ$ is strictly more than $(2n+m)\alpha + n + \nfrac{3mn}{4}$ since $\alpha<0$. This contradicts the assumption that $c$ co-wins in $\bar{\PP}\cup\QQ$. Hence $\{x_i^*\}_{i\in[n]}$ is a satisfying assignment of the clause in \TT and thus \II is a \YES instance.
\end{proof}

\MaximinPoly*

\begin{proof}
Let the input instance of \PW be $(\CC, \PP, c)$ where every partial vote in \PP has at most one pair of candidates whose ordering is undetermined. We consider an extension $\PP^\pr$ of \PP where the candidate $c$ is placed as high as possible. Notice that the maximin score of $c$ in every extension of $\PP^\pr$ is same and known since the relative ordering of $c$ with other candidates is fixed in $\PP^\pr$. Let the maximin score of $c$ in $\PP^\pr$ be $s(c)$. We now observe that, if $c$ is a weak Condorcet winner, that is $s(c)\ge 0$, then $c$ is a co-winner in every extension of $\PP^\pr$ and thus $(\CC, \PP, c)$ is a \YES instance. Otherwise, let us assume $s(c)<0$. For any two candidates $x, y\in\CC\setminus\{c\}$, let $\VV_{x,y}$ be the set of partial votes in $\PP^\pr$ where the ordering between the candidates $x$ and $y$ is undetermined. Since every partial vote in $\PP^\pr$ has at most one undetermined pair, for every $x_1, x_2, y_1, y_2\in\CC\setminus\{c\}$, $\VV_{x_1,x_2} \cap \VV_{y_1, y_2} = \emptyset$. 

We construct the following flow graph \GG. We have a vertex for every subset $\{x,y\}\subseteq\CC\setminus\{c\}$ of candidates other than $c$ of size two, a vertex for every candidate other than $c$, and two special vertces $s$ and $t$. If making $x$ prefer over $y$ in every $\VV_{x,y}$ makes $\DD(x,y) < s(c)$, then we add a directed edge from the vertex $\{x,y\}$ to $x$ of capacity one. Similarly, if making $y$ prefer over $x$ in every $\VV_{x,y}$ makes $\DD(y,x) < s(c)$, then we add a directed edge from the vertex $\{x,y\}$ to $y$ of capacity one. We add an edge of capacity one from $s$ to every vertex corresponding to the vertex $\{x,y\}$ for every $\{x,y\}\subseteq\CC\setminus\{c\}$. Let $\bar{\CC}\subseteq\CC\setminus\{c\}$ be the set of candidates $x$ in $\CC$ such that there exists a candidate $y\in\CC\setminus\{x\}$ such that $\DD(x,y) \le s(c)$ in every extension of $\PP^\pr$; note that this is easy to check by guessing the candidate $y$. We add an edge of capacity one from the vertex corresponding to every candidate in $\CC\setminus\bar{\CC}$ to $t$. We claim that the $(\CC, \PP^\pr, c)$ is a \YES instance if and only if there is a $s-t$ flow in \GG of amount $|\CC\setminus\bar{\CC}|$.

Suppose $(\CC, \PP^\pr, c)$ is a \YES instance. Consider an extension $\bar{\PP}$ of $\PP^\pr$ where $c$ co-wins. Let the extension of $\VV_{x,y}$ in $\bar{\PP}$ be $\bar{\VV}_{x,y}$. For every candidate $x\in\CC\setminus\bar{\CC}$, there exists a candidate $y\in\CC\setminus\{c\}$ such that $\DD(x,y)\le s(c)$; we call that candidate $d(x)$ (if there are more than one such $y$, we pick any one). Then we assign a flow of unit one along the path $s\rightarrow \{x,d(x)\}\rightarrow x\rightarrow t$. For any two candidates $x,y\in\CC\setminus\bar{\CC}$, since at most one of $\DD(x,y)$ and $\DD(y,x)$ be less than $0$ (and thus at most one of them can be $\le s(c)$), we never assign flows to both the paths $s\rightarrow \{x,y\}\rightarrow x\rightarrow t$ and $s\rightarrow \{x,y\}\rightarrow y\rightarrow t$. Hence the flow is valid. Since every vertex in $\CC\setminus\bar{\CC}$ sends exactly one unit of flow to $t$, the total amount of flow in \GG is $|\CC\setminus\bar{\CC}|$.

In the reverse direction, suppose there exists a $s-t$ flow $f$ in \GG of amount $|\CC\setminus\bar{\CC}|$. Then, for every vertex $x\in\CC\setminus\bar{\CC}$, there exists a $d(x)\in\CC\setminus\{c,x\}$ such that $f$ assigns a one unit of flow from the vertex $\{x,d(x)\}$ to $x$. For every candidate $x\in\CC\setminus\bar{\CC}$, we make $d(x)\suc x$ in the completion of every vote in $\VV_{x,d(x)}$. We fix the ordering of all other pairs of candidates arbitrarily. We use $\bar{\VV}$ to denote the resulting completion of $\PP^\pr$. By construction of \GG, $c$ is a co-winner in $\bar{\VV}$.
\end{proof}

\BucklinPoly*

\begin{proof}

We first establish the hardness result. The \PW problem for the Bucklin voting rule is clearly in \NP. To prove \NP-hardness of \PW, we reduce \PW from \TDM. Let $\II = (\XX\cup\YY\cup\ZZ, \SS)$ be an arbitrary instance of \TDM. Let $|\XX|=|\YY|=|\ZZ|= m$, $|\SS|=t$, and $\UU = \XX\cup\YY\cup\ZZ$. For every $a\in\UU$, let $f_a$ be the number of sets in \SS where $a$ belongs, that is $f_a = |\{\sss\in\SS : a\in\sss\}|$. We assume, without loss of generality, that $f_a\le 3$ for every $a\in\UU$ since \TDM is \NPC even with this restriction \cite{kann1991maximum}. We also assume without loss of generality that $t>3m$ (otherwise we duplicate the sets in \SS). We construct an instance $\II^\pr$ of \PW from \II as follows.
 \[\text{Set of candidates: } \CC = \XX\cup\YY\cup\ZZ\cup\{c\}\cup\GG_1\cup\GG_2\cup\GG_3, \text{ where } |\GG_1| = |\GG_2| = |\GG_3| = 3m \]
 
 For every $\sss=(x, y, z)\in\SS$, let us consider the following vote $\ppp_\sss$. 
 \[ \ppp_\sss = (\UU\setminus\{x,y,z\}) \suc x \suc y \suc z \suc \text{others} \]
 
 Using $\ppp_\sss$, we define a partial vote $\ppp_\sss^\pr$ as follows. 
 \[ \ppp_\sss^\pr = \ppp_\sss \setminus \{(x,y), (x,z)\} \]
 
 Let us define $\PP = \cup_{\sss\in\SS} \ppp_\sss$ and $\PP^\pr = \cup_{\sss\in\SS} \ppp_\sss^\pr$. For $i\in[3]$ and $j\in[3m]$, let $\GG_i^j$ denote an arbitrary subset of $\GG_i$ of size $j$. We add the following set \QQ of complete votes as in \Cref{tbl:more_votes_bucklin}.
%
 \begin{table}[!htbp]
 \centering
 \resizebox{\textwidth}{!}{
  \begin{tabular}{|c|c|}\hline
   \makecell{$\forall z\in\ZZ$, $f_z-1$ copies of $c\suc \GG_1^{3m-4}\suc z \suc \text{others}$\\$1$ copy of $c\suc \GG_1^{3m-3}\suc z\suc \text{others}$}&$\forall y\in\YY$, $f_y$ copies of $\GG_1^{3m-3}\suc y\suc c\suc \text{others}$  \\
   $\forall x\in\XX$, $3$ copies of $(\XX\setminus\{x\})\suc\YY\suc\GG_2^{m-1} \suc x\suc \text{others}$ & $t-3m$ copies of $\XX\suc\YY\suc \GG_2^m\suc \text{others}$\\
   $t-1$ copies of $\ZZ\suc \GG_2^{2m} \suc \text{others}$ & $1$ copy of $\ZZ\suc \XX \suc \GG_2^m \suc \text{others}$\\
   $t$ copies of $c\suc \XX\suc \YY\suc \GG_2^m \suc \text{others}$ & $t-1$ copies of $c\suc \YY \suc \ZZ \suc \GG_3^{m} \suc \text{others}$\\
   $1$ copy of $c\suc \ZZ\suc \XX \suc \GG_3^{m} \suc \text{others}$ & $t-2$ copies of $\ZZ\suc \XX\suc \GG_3^m \suc \text{others}$\\
   $2$ copies of $c\suc \ZZ\suc \XX \suc \GG_3^m \suc \text{others}$ & $1$ copy of $\ZZ\suc \XX\suc \GG_3^m \suc \text{others}$\\\hline
  \end{tabular}}
  \caption{We add the following set of complete votes \QQ.}\label{tbl:more_votes_bucklin}
 \end{table}

 We summarize the number of times every candidate gets placed within top $3m-1$ and $3m-2$ positions in $\PP \cup \QQ$ in \Cref{tbl:buck_initial}. We now define the instance $\II^\pr$ of \PW to be $(\CC, \PP^\pr \cup \QQ, c)$. The total number of votes in $\II^\pr$ is $8t+1$. Notice that the number of undetermined pairs of candidates in every vote in $\II^\pr$ is at most $2$. This finishes the description of the \PW instance. We claim that \II and $\II^\pr$ are equivalent. 
 \begin{table}[!htbp]
  \centering
  \begin{tabular}{|ccc|}\hline
   Candidates & Top $3m-1$ positions & Top $3m-2$ positions\\\hline\hline
   $c$ & $4t+2$ & $3t+2$\\
   $x\in\XX$ & $4t+3$ & $4t$\\
   $y\in\YY$ & $\le 4t+2$ & $4t-1$\\
   $z\in\ZZ$ & $4t+1$ & $4t$\\
   $g\in\GG_1\cup\GG_2\cup\GG_3$ & $<4t$ & $<4t$\\\hline
  \end{tabular}
  \caption{Number of times every candidate is initially placed within top $3m-1$ and $3m-2$ positions.}\label{tbl:buck_initial}
 \end{table}
  
 In the forward direction, suppose that \II be a \YES instance of \TDM. Then there exists a collection of $m$ sets $\SS^\pr\subset\SS$ in \SS such that $\cup_{\AA\in\SS^\pr} \AA = \XX\cup\YY\cup\ZZ$. We extend the partial vote $\ppp_\sss^\pr$ to $\bar{\ppp}_\sss$ as follows for $\sss\in\SS$. 
 $$
 \bar{\ppp}_\sss = 
 \begin{cases}
  (\UU\setminus\{x,y,z\}) \suc y \suc z \suc x \suc \text{others} & \sss\in\SS^\pr\\
  (\UU\setminus\{x,y,z\}) \suc x \suc y \suc z \suc \text{others} & \sss\notin\SS^\pr
 \end{cases}
 $$
 
 We consider the extension of \PP to $\bar{\PP} = \cup_{\sss\in\SS} \bar{\ppp}_\sss$. We claim that $c$ is a co-winner in the profile $\bar{\PP}\cup\QQ$ since $c$ gets majority within top $3m-1$ positions with $4t+2$ votes, whereas no candidate gets majority within top $3m-2$ positions and every candidate in \CC is placed at most $4t+2$ times within top $3m-1$ positions.
 
 In the reverse direction suppose the \PW instance $\II^\pr$ be a \YES instance. Then there exists an extension of the set of partial votes $\PP^\pr$ to a set of complete votes $\bar{\PP}$ such that $c$ is a co-winner in $\bar{\PP}\cup\QQ$. Let us call the extension of $\ppp_\sss^\pr$ in $\bar{\PP}$ $\bar{\ppp}_\sss$. First we notice that for $c$ to co-win, every $x\in\XX$ must be placed at positions outside top $3m-1$ since otherwise $x$ will receive more votes that $c$ within top $3m-1$ positions. Also observe that the only way to place $x$ outside the top $3m-1$ positions in the votes in $\bar{\ppp}_\sss$ for some $\sss=(x,y,z)$ is to put $x, y$ and $x$ at $3m$, $3m-2$, and $3m-1$ positions respectively. We consider the subset $\SS^\pr\subseteq\SS$ of \SS whose corresponding vote completions place $x$ at $3m^{th}$ position; that is $\SS^\pr = \{ \sss=(x,y,z)\in\SS: \bar{\ppp}_\sss = (\UU\setminus\{x,y,z\}) \suc y \suc z \suc x \suc \text{others} \}$. From the discussion above, we have $|\SS^\pr|\ge m$. Now we observe that every $y\in\YY$ can be placed at most once at the $(3m-2)^{th}$ position in the votes in $\{\bar{\ppp}_\sss: \sss\in\SS^\pr\}$; otherwise $y$ will get majority within top $3m-2$ positions and $c$ cannot win the election.
 Also every $z\in\ZZ$ can be placed at most once at the $(3m-1)^{th}$ position in the votes in $\{\bar{\ppp}_\sss: \sss\in\SS^\pr\}$; otherwise $z$ will receive strictly more than $4t+2$ votes within top $3m-1$ positions and thus $c$ cannot win. Hence, every $x\in\XX, y\in\YY$ and $z\in\ZZ$ belong in exactly one set in $\SS^\pr$ and thus $\SS^\pr$ forms a three dimensional matching. Hence \II is a \YES instance.

 We now turn to the polynomially solvable scenario claimed in the theorem. Let the input instance of \PW be $(\CC, \PP, c)$ where every partial vote in \PP has at most one pair of candidates whose ordering is undetermined. We consider an extension $\PP^\pr$ of \PP where the candidate $c$ is placed as high as possible. Let $k$ be the minimum integer such that $c$ gets majority within top $k$ positions. Now we can use the polynomial time algorithm for the $k$-approval voting rule to solve this instance.
\end{proof}

\end{document}